\documentclass[11pt]{article}

\usepackage{amsmath,amssymb,amsthm,amsfonts}
\usepackage{latexsym,bbm,xspace,graphicx,float,mathtools}
\usepackage{amsthm,amsfonts,amssymb,graphicx,float,graphics}
\usepackage{fullpage}
\usepackage{caption}
\usepackage{subcaption}
\usepackage[usenames,dvipsnames]{xcolor}
\usepackage{amsmath}
\usepackage{setspace}
\usepackage[backref, colorlinks,citecolor=blue,bookmarks=false]{hyperref}

\spacing{1.04}

\def\SUM{\mathsf{SUM}}

\newtheorem{theorem}{Theorem}
\newtheorem{lemma}{Lemma}[section]
\newtheorem{claim}[lemma]{Claim}

\newtheorem{fact}[lemma]{Fact}

\theoremstyle{definition}

\newtheorem{remark}{Remark}

\theoremstyle{plain}

\newcommand{\E}{\operatorname{{\bf E}}}
\newcommand{\Ex}{\mathop{{\bf E}\/}}
\renewcommand{\Pr}{\operatorname{{\bf Pr}}}

\newcommand{\poly}{\mathrm{poly}}

\newcommand{\supp}{\mathrm{supp}}

\newcommand{\eps}{\varepsilon}
\renewcommand{\epsilon}{\eps}

\newcommand{\Sipser}{\mathsf{Sipser}}

\newcommand{\OR}{\mathsf{OR}}
\newcommand{\AND}{\mathsf{AND}}
\newcommand{\MOD}{\mathsf{MOD}}
\newcommand{\NOT}{\mathsf{NOT}}
\newcommand{\THR}{\mathsf{THR}}

\newcommand{\bS}{\mathbf{S}}

\newcommand{\bX}{\mathbf{X}}
\newcommand{\bY}{\mathbf{Y}}

\newcommand{\bt}{\boldsymbol{T}}

\newcommand{\bx}{\boldsymbol{x}}
\newcommand{\by}{\boldsymbol{y}}
\newcommand{\bz}{\boldsymbol{z}}

\newcommand{\bGamma}{{\mathbf{\Gamma}} }

\newcommand{\calD}{\mathcal{D}}

\newcommand{\calS}{\mathcal{S}}

\makeatletter
\newtheorem*{rep@theorem}{\rep@title}
\newcommand{\newreptheorem}[2]{
\newenvironment{rep#1}[1]{
 \def\rep@title{#2 \ref{##1}}
 \begin{rep@theorem}\itshape}
 {\end{rep@theorem}}}
\makeatother

\theoremstyle{plain}



\def\colorful{1}

\ifnum\colorful=1

\fi
\ifnum\colorful=0

\fi

\newcommand{\ignore}[1]{}

\newreptheorem{theorem}{Theorem}
\newtheorem*{theorem*}{Theorem}
\newreptheorem{lemma}{Lemma}
\newreptheorem{proposition}{Proposition}
\newtheorem*{noclaim*}{Claim}

\def\calD{\mathcal{D}} \def\zz{\mathbf{z}}

\newcommand{\uhr}{\upharpoonright}

\newcommand{\eqdef}{\stackrel{\rm def}{=}}
\def\yesd{\mathcal{YES}} \def\nod{\mathcal{NO}}
\def\11{\mathbf{1}} \def\00{\mathbf{0}}
\def\xx{\bx}
\def\yy{\by}
\def\zz{{\bz}}
\def\rr{\boldsymbol{R}}

\begin{document}

\thispagestyle{empty}

\title{Addition is exponentially harder than counting\\
for shallow monotone circuits}

\vspace{+6ex}

\author{Xi Chen\thanks{xichen@cs.columbia.edu.}
 \qquad Igor C. Oliveira\thanks{oliveira@cs.columbia.edu.} \qquad Rocco A. Servedio\thanks{rocco@cs.columbia.edu.  Supported in part by NSF grants CCF-1319788 and CCF-1420349.}\\~\\\small{Department of Computer Science}\\\small{Columbia University}\\~\\}

\maketitle

\vspace{-3ex}

\begin{abstract}
Let $U_{k,N}$ denote the Boolean function which takes as input $k$  strings of $N$ bits each, representing $k$ numbers $a^{(1)},\dots,a^{(k)}$ in $\{0,1,\dots,2^{N}-1\}$, and outputs 1 if and only if $a^{(1)} + \cdots + a^{(k)} \geq 2^N.$   Let $\THR_{t,n}$ denote a \emph{monotone unweighted threshold gate}, i.e., the Boolean function which takes as input a single string $x \in \{0,1\}^n$ and outputs $1$ if and only if $x_1 + \cdots + x_n \geq t$. The \mbox{function} $U_{k,N}$ may be viewed as a monotone function that performs addition, and $\THR_{t,n}$ may be viewed as a monotone function that performs counting. We refer to circuits that are composed of $\THR$ gates as \emph{monotone majority 
  circuits.}

The main result of this paper is an exponential lower bound on the size of bounded-depth 
monotone majority circuits that compute $U_{k,N}$.   More precisely, we show that for any constant $d \geq 2$, any depth-$d$   monotone majority circuit computing $U_{d,N}$ must\vspace{0.008cm} have size $\smash{2^{\Omega(N^{1/d})}}$.  Since $U_{k,N}$ can be computed by a single  monotone \emph{weighted} threshold gate (that uses exponentially large weights), our lower bound implies that constant-depth monotone majority circuits require exponential size to simulate monotone weighted~threshold gates. This answers a question posed by Goldmann and Karpinski (STOC'93) and recently restated by H\aa stad (2010, 2014).  \vspace{0.01cm}We also show that our lower bound is essentially best possible, by constructing a depth-$d$, size-$\smash{2^{O(N^{1/d})}}$ 
monotone majority circuit
for $U_{d,N}$.

As a corollary of our lower bound, we significantly strengthen a classical theorem in circuit complexity due to Ajtai and Gurevich (JACM'87). They exhibited a monotone function that~is in $\mathsf{AC}^0$ but requires super-polynomial size for any constant-depth monotone circuit composed~of unbounded fan-in $\AND$ and $\OR$ gates.  We describe a monotone function that is in depth-3 $\mathsf{AC}^0$ but requires \emph{exponential} size monotone circuits of any constant depth, even if the circuits are composed of $\THR$ gates.

\end{abstract}

 \thispagestyle{empty}
\spacing{1.03}

\newpage
\setcounter{page}{1}

% %%%%%%%%%%%%%%%%%%%%%%%%%%%%%%%%
% intro
% %%%%%%%%%%%%%%%%%%%%%%%%%%%%%%%%

\newpage

\section{Introduction.}\label{s:introduction}\vspace{0.05cm}

\begin{quote}
{\em{``And you do Addition?'' the White Queen asked. ``What's one and one and
one and one and one and one and one and one and one and one?''

``I don't know,'' said Alice. ``I lost count.''

``She can't do Addition,'' the Red Queen interrupted.}}

\hskip 2.8in~--- Lewis Carroll,  \em{Through the Looking Glass\vspace{0.06cm}}

\end{quote}

\noindent {\bf Threshold functions and threshold circuits.}
A Boolean function $f\colon \{0,1\}^n \to \{0,1\}$ is called a \emph{weighted threshold function} (also known as a halfspace, weighted majority, weighted threshold gate, or linear threshold function) if there exist integers $w_1, \ldots, w_n$ and $t$ such that
$$
f(x) = 1 \quad \Longleftrightarrow \quad \sum_{i = 1}^n w_i x_i \geq t.
$$
The parameters $w_1, \ldots, w_n$ are called \emph{weights}.
We say that a threshold function $f$ is \emph{unweighted} if $|w_i| = 1$ for every $i \in \{1, \ldots, n\}$, and that it is \emph{monotone} if every weight is non-negative.  (Thus a monotone unweighted threshold function is precisely a $\THR_{t,n}$ function described in the abstract.)

Threshold functions and their generalizations have been extensively investigated for decades (see e.g. Dertouzos \cite{Dertouzos}, Minsky and Papert \cite{DBLP:books/daglib/0066902}, and Muroga \cite{DBLP:books/daglib/0098527}), and arise in diverse areas including social choice theory (Taylor and Zwicker \cite{MR1092927}), circuit complexity (Aspnes et al. \cite{DBLP:journals/combinatorica/AspnesBFR94}), structural complexity (Beigel, Reingold, and Spielman \cite{DBLP:journals/jcss/BeigelRS95}), learning theory (Freund and Schapire \cite{DBLP:journals/jcss/FreundS97}), neural networks (Parberry \cite{Parberry}), cryptography (Naor and Reingold \cite{DBLP:journals/jacm/NaorR04}), and many others.

In this work, we consider Boolean circuits that are composed of gates that compute
  threshold functions (i.e., \emph{threshold gates}). (We refer to Jukna \cite{DBLP:books/daglib/0028687} as an extensive reference on Boolean functions and circuit complexity). 
While individual threshold gates may appear relatively simple, Boolean circuits composed of these gates (i.e., \emph{threshold circuits}) remain poorly understood~\mbox{despite} intensive study. For instance, it is a notorious and long-standing open problem in complexity theory~to prove the existence of a function in $\mathsf{NP}$ that cannot be computed by
  a depth-2 \mbox{circuit} with polynomially many weighted threshold gates. This difficulty can be explained in part by the surprising
  computational power of bounded-depth threshold circuits, both in theory and practice. On the theory side, such circuits~can efficiently implement all the basic arithmetic operations (see e.g., Table 1 in Sherstov \cite{DBLP:journals/ipl/Sherstov07}) and~can also simulate (in quasi-polynomial size and depth 3)  $\AND/\OR/\MOD_m$ Boolean circuits of much larger depth (Allender \cite{DBLP:conf/focs/Allender89} and Yao \cite{DBLP:conf/focs/Yao90}). On a more practical level, constant-depth networks of (continuous analogues of) threshold gates play a fundamental role in recent successful deep learning frameworks (see e.g., Schmidhuber \cite{DBLP:journals/nn/Schmidhuber15}). 

Despite our inability to prove strong lower bounds against threshold circuits, there have been some notable successes in understanding the relative power of weighted versus unweighted threshold~gates and circuits. Siu and Bruck \cite{DBLP:journals/siamdm/SiuB91} were the first to show that any weighted threshold gate can be simulated by a polynomial-size, constant-depth circuit consisting of unweighted threshold gates (such circuits are also known as \emph{majority circuits}). This result was improved by Goldmann, H\aa stad, and Razborov in \cite{DBLP:journals/cc/GoldmannHR92}, who showed (non-constructively) that weighted threshold gates can be computed by polynomial-size majority circuits of depth 2; in fact, \cite{DBLP:journals/cc/GoldmannHR92} showed that~any depth-$d$ weighted threshold circuit can be simulated efficiently by a depth-$(d+1)$ majority circuit. Soon thereafter Goldmann and Karpinski \cite{DBLP:conf/stoc/GoldmannK93} gave a constructive proof with better parameters for the size of the resulting majority circuits. Subsequent simplifications and improvements of these simulations were given by Hofmeister \cite{DBLP:conf/cocoon/Hofmeister96} and Amano and Maruoka \cite{DBLP:conf/mfcs/AmanoM05}.

\medskip

\noindent {\bf Monotone functions and monotone circuits.}  In a different, and highly successful, strand~of circuit complexity research, a wide range of lower bounds have been obtained against various types of \emph{monotone} Boolean circuits (composed of $\AND/\OR$ gates only but no negations).  A sequence of well-known results \cite{razborov1985Clique,andreev1985method,DBLP:journals/combinatorica/AlonB87,DBLP:journals/combinatorica/Tardos88} culminated in the existence of explicit monotone Boolean functions that can be computed by polynomial-size Boolean circuits~but require \emph{monotone} circuits of exponential size.  Analogous results highlighting the limitations of monotone circuits are also known at the ``low-complexity'' end of the spectrum: in an important result, 
Ajtai and Gurevich \cite{DBLP:journals/jacm/AjtaiG87} exhibited a monotone function in $\mathsf{AC}^0$ (i.e., a constant-depth, \mbox{polynomial-size} $\AND/\OR/\NOT$ Boolean circuit) that requires \emph{monotone} $\mathsf{AC}^0$ circuits 
(composed of $\AND/\OR$ gates) 
to have super-polynomial size.  However, it should be noted that the Ajtai--Gurevich circuit lower bound against monotone $\mathsf{AC}^0$ is quantitatively not very strong (at best a
  quasipolynomial~$n^{\Omega(\log n)}$ lower bound; see discussion following the statement of
  the Ajtai--Gurevich theorem below).   Other works have given alter\-native/simplified expositions of the Ajtai--Gurevich lower bound and of its consequences in formal logic (see \cite{BlaisSchederTan:13} for the former and Stolboushkin \cite{Stolboushkin:95} for the latter). But prior to the results of this paper,  stronger lower bounds against monotone $\mathsf{AC}^0$ circuits for monotone functions in $\mathsf{AC}^0$ remained elusive.

\medskip

\noindent {\bf This work:  Monotone weighted threshold functions versus constant-depth monotone majority circuits.}   As mentioned earlier, Goldmann and Karpinski  gave a constructive proof \cite{DBLP:conf/stoc/GoldmannK93} that weighted threshold gates can be simulated by polynomial-size and depth-2 majority circuits. They also observed that even if the weighted threshold gate is monotone, known simulations produce majority circuits that are inherently \emph{non-monotone} (i.e., they contain majority gates with negative weights, or equivalently, negation gates), which then led them to ask the question of whether an efficient monotone simulation is possible in constant depth.  

Hofmeister \cite{DBLP:conf/coco/Hofmeister92} made some early progress on this question by 
showing that any monotone depth-2 majority circuit that computes the function $U_{2,N}$
  from the abstract must have exponential size. To~state the result more precisely, let us   
  first clearly specify our notion of monotone majority circuits.
A monotone majority circuit here
  is a directed acyclic graph which may have multiple edges
  (called wires). 
There is~a single node with no outgoing wires, called the output gate.
Nodes that have no incoming wires~are called input nodes and are 
  each labeled either $0$, $1$ or $x_i$, for some $i$;
every other node is labeled with a 
  monotone unweighted threshold gate $\THR_{t,m}$ for some $t$, 
  with $m$ being its in-degree,
  which outputs $1$ iff there are at least $t$ $1$'s from its
  $m$ input wires.
We~say the \emph{size} of a monotone unweighted threshold gate $\THR_{t,m}$ is $m$
  (or its in-degree), and that the size of 
  a monotone majority circuit is the sum of the sizes of its gates
  (or its number of wires).\footnote{Observe that by reduplicating inputs, any weighted threshold function {$f$ given by} $\sum_{i=1}^n w_i x_i \geq t$ can be computed by an unweighted threshold gate of size $|w_1| + \cdots + |w_n|.$ {We sometimes refer to this as the ``weight of $f$.''}} Then\vspace{0.008cm} Hofmeister 
    showed that every depth-2 monotone majority circuit for $U_{2,N}$ must have size $\smash{2^{\Omega(\sqrt{N})}}$.
  
As mentioned above, in subsequent work  \cite{DBLP:conf/cocoon/Hofmeister96}
  and \cite{DBLP:conf/mfcs/AmanoM05}, several improvements were made on the Goldmann-Karpinski simulation, but neither
  is monotone, and no further progress was~obtained on the lower bound side after Hofmeister's paper \cite{DBLP:conf/coco/Hofmeister92} until the current work.   The question of Goldmann and Karpinski was recently restated by H\aa stad \cite{Hastad:10,DBLP:journals/dagstuhl-reports/BeyersdorffHKS14}.

\subsection{Our Results.}

Our main result shows that monotone weighted threshold gates cannot be simulated by subexponential size monotone majority circuits of constant depth.  This may be viewed as an extension~of Hofmeister's depth-2 lower bound in \cite{DBLP:conf/coco/Hofmeister92} to arbitrary constant depth (in fact we obtain super-polynomial size lower bounds even for circuits of small super-constant depth; see discussions after Theorem~\ref{thm:lower} below). We thus answer the question posed by Goldmann and Karpinski \cite{DBLP:conf/stoc/GoldmannK93} and~by H\aa stad \cite{Hastad:10,DBLP:journals/dagstuhl-reports/BeyersdorffHKS14}.

Before giving a precise statement of our results, we define formally the family $U_{k,N}$ of   Boolean functions as described in the abstract. Given $t\ge 1$, we let $[t]$ denote the set $\{1,\dots,t\}$.  {For $k \geq 2$, the function} $U_{k,N}$ maps $\{0,1\}^{k\times N}$ to $\{0,1\}$ as follows.   Given $x=(x_{i,j})_{i \in [k],j \in [N]} \in  \{0,1\}^{k\times N}$, define
$$
\SUM(x)\eqdef \sum_{j=1}^{N} 2^{N-j}\cdot (x_{1,j} + \cdots + x_{k,j}) \quad\quad  \text{and} \quad\quad 
U_{k,N}(x) \eqdef \begin{cases}
1 & \text{if~} \SUM(x)\ge 2^N,\\0 & \text{otherwise}.
\end{cases}
$$
It is helpful to think of the input $x=(x_{i,j})_{i \in [k],j \in [N]}$ as a $k$-row, $N$-column, and $0/1$-valued~matrix, where its $i$th row $(x_{i,1},\dots,x_{i,N})$ gives the binary representation of a number $\smash{x^{(i)}\in \{0,1,\dots,2^{N}-1\}}$ in the usual way (with $x_{i,1}$ being the most significant bit).  Then the function $U_{k,N}$ adds up the $k$ numbers 
  $x^{(1)},\ldots,x^{(k)}$ and outputs 1 if and only if the sum is at least $2^N.$
	
	With the definition of $U_{k,N}$ in place, our main result can be stated as follows:

\begin{theorem} \label{thm:lower}
Let $d \geq 2$, $n$ and $N$ be three positive integers that satisfy $$n\ge 2^{{60d}}
\quad\text{and}\quad
  N\ge ({2^{13}} n)^d.$$
Then any depth-$d$ monotone majority circuit that computes $U_{d,N}$ must 
  have size at least $2^{n/2^{60d}}.$
  \end{theorem}

This lower bound is nearly optimal for any fixed $d\ge 2$, as we prove the following upper bound.

\begin{theorem} \label{thm:upper}
Let {$k,d,N \geq 2$ be three positive integers.} Then there exists a depth-$d$ monotone\\ majority circuit of size $2^{6 (N^{1/d} \log k + \log N)}$ that computes $U_{k,N}$.
\end{theorem}

\begin{remark}
For any fixed constant $d\ge 2$,
  Theorems  \ref{thm:lower} and \ref{thm:upper} together
  show that {the smallest} depth-$d$ monotone majority circuit that computes
  $U_{d,N}$ (note that this function has $d \cdot N=\Theta(N)$ input variables) has size $\exp(\Theta(N^{1/d}))$.
In addition, by setting $d=c\sqrt{\log N}$ and $n = 2^{61d}$ for some small enough positive constant $c$
  so that $N\ge (2^{13}n)^d$,
  Theorem \ref{thm:lower} implies that any depth-$d$ monotone majority circuit
  computing $U_{d,N}$ has superpolynomial size (exponential in
  $\smash{2^{c\sqrt{\log N})}}$).
\end{remark}

\begin{remark}
As an easy consequence of Theorem~\ref{thm:upper}, we obtain a slightly weaker version of the main result of Beimel and Weinreb \cite{DBLP:conf/coco/BeimelW05}.  They proved that the ``universal monotone~threshold function''\footnote{It is called the universal monotone threshold function because
  it can simulate any monotone weighted threshold function over $N$ inputs.}  $\smash{U_{O(N), O(N \log N)}}$ can be computed by a $\poly(N)$-size, depth-$O(\log N)$ monotone circuit 
  composed of fan-in two $\AND$ gates and unbounded fan-in $\OR$ gates. While Theorem~\ref{thm:upper} above is tailored for small values of $k$, we note that it implies that $\smash{U_{O(N), O(N \log N)}}$ can be computed by a
  $\poly(N)$-size, depth-$O(\log^2 N)$ monotone circuit composed of 
  fan-in two $\AND/\OR$ gates only. (In more detail, it is enough to set $d=\log N$ and replace each majority gate by a $O(\log N)$-depth fan-in-two AND/OR Boolean circuit.) We sketch a simpler construction in Appendix \ref{s:connectivity} that matches the parameters obtained in \cite{DBLP:conf/coco/BeimelW05} in the case of the universal monotone threshold function. 
  \end{remark}

Another  consequence of our lower bound as stated in Theorem \ref{thm:lower} is a significant strengthening of the Ajtai--Gurevich lower bound discussed earlier.  We recall their result in more detail:

\begin{theorem*}[Ajtai--Gurevich \cite{DBLP:journals/jacm/AjtaiG87}] \label{p:original_AG}
There exists an explicit  sequence $f = \{f_n\}_{n \in \mathbb{N}}$ of monotone Boolean functions $f_n \colon \{0,1\}^n \to \{0,1\}$ such that:
\begin{enumerate}
\item[\emph{(}i\emph{)}] $f \in \mathsf{AC}^0$;\vspace{-0.16cm}
\item[\emph{(}ii\emph{)}] $f \notin \mathsf{monAC}^0$: For any fixed constant $d$, any monotone depth-$d$ $\AND/\OR$ circuit computing $f_n$ must have size at least $S_d(n)$, for some function $S_d(n)= n^{\omega(1)}$.
\end{enumerate}
\end{theorem*}

\noindent Regarding part (\emph{ii}) above, it is not immediately clear what is the best (largest) function $S_d(n)$ that can be extracted from the Ajtai--Gurevich proof. However, $f_n$ is easily seen to be computed by a monotone depth-$2$ circuit (a monotone DNF) of size $n^{\log n}$, so $S_d(n) \leq n^{\log n}$ for all $d \geq 2.$

As an easy corollary of Theorem~\ref{thm:lower}, we  strengthen the Ajtai--Gurevich 
  circuit lower bound (for~a different monotone function  in $\mathsf{AC}^0$) in two ways:  by giving a lower bound against monotone \emph{majority~circuits} of constant depth (rather than monotone circuits of $\AND/\OR$ gates only), and by achieving an \emph{exponential} size lower bound for any fixed depth (rather than a bound which is at most $n^{\log n}$).  
  Our theorem is the following:
\begin{theorem} \label{t:strong_AG}
There exists an explicit sequence $g = \{g_n\}_{n \in \mathbb{N}}$ of monotone Boolean functions, where $g_n \colon \{0,1\}^{n \log n}$ $\to \{0,1\}$, such that:
\begin{enumerate}
\item[\emph{(}i\emph{)}] $g \in \mathsf{AC}^0$ \emph{(}in fact each $g_n$ is computed by a $\poly(n)$-size, depth-$3$ $\AND/\OR/\NOT$ circuit\emph{)};\vspace{-0.16cm}
\item[\emph{(}ii\emph{)}] For any constant $d \geq 2$, any monotone depth-$d$ majority circuit for $g_n$ must have size $2^{\Omega(n^{1/d})}$.
\end{enumerate}
\end{theorem} 

It is interesting to observe that our proof of Theorem \ref{t:strong_AG} uses very different arguments from those of
Ajtai and Gurevich.  The heart of their proof is a ``switching lemma'' for monotone functions on hypergrids (see the excellent exposition of their proof given in \cite{BlaisSchederTan:13}), whereas our approach does not use switching
lemmas at all.

\subsection{Related Work and Our Techniques.}

In addition to papers discussed above, the works of Yao \cite{DBLP:conf/stoc/Yao89} and H\aa stad and Goldmann \cite{DBLP:journals/cc/HastadG91} are relevant in the context of our lower bound result.  Let $\Sipser_{d+1}$ denote the read-once monotone $n$-variable formula of depth $d+1$ that has alternating layers of $\AND$ and $\OR$ gates (see \cite{DBLP:journals/cc/HastadG91} for a detailed description of this function). Strengthening the earlier result of Yao \cite{DBLP:conf/stoc/Yao89}, H\aa stad and Goldmann \cite{DBLP:journals/cc/HastadG91} showed that a depth-$d$ circuit of \emph{weighted} monotone threshold gates computing $\Sipser_{d+1}$ must have size $\smash{2^{\Omega(n^{1/2d})}}.$  In contrast, our Theorem~\ref{thm:lower} only establishes a lower bound against constant-depth monotone circuits of \emph{unweighted} threshold gates, but --- crucially --- we establish the lower bound for a much ``simpler'' monotone function, $U_{d,N}$, that is computed by a \emph{single} weighted monotone threshold gate.   Indeed, the main challenge of our work is to push through a lower bound for such a heavily constrained target function.

At the heart of our lower bound proof is a sequence of carefully constructed pairs of probability distributions $(\yesd_\ell,\nod_\ell)$ over $\{0,1\}^{(\ell+ {1})\times N_\ell}$
  for $\ell= {1},\ldots,d$ (i.e. over possible inputs to $U_{\ell+ {1}, N_\ell}$ for some 
  $N_\ell$ to be specified later).
The first distribution $\yesd_\ell$ in the pair is supported on strings $x$ that have $U_{\ell+ {1},N_\ell}(x)=1$, while $\nod_\ell$ is supported on strings with $U_{\ell+ {1}, N_\ell}(x)=0.$  The key property~of these pairs of distributions, which yields our lower bound, is that considered together, each pair of $(\yesd_\ell,\nod_\ell)$  is ``hard'' for ``small'' monotone majority circuits of depth $\ell$ in a suitable sense.  In a bit more detail, our requirement is roughly that for any such circuit $F$, we have
\begin{equation} \label{eq:template}
\Pr_{\xx \sim \yesd_{\ell}}\big[F(\xx)=1\big] + 
\Pr_{\yy \sim \nod_{\ell}}\big[F(\yy)=0\big]  \leq 1 + \tau_\ell,
\end{equation}
for a suitable value $0 < \tau_\ell \ll 1$.  
At a high level, we establish (\ref{eq:template}) above through a careful inductive argument on $\ell$.  (We note that the preceding sketch is something of an oversimplification; actually, in order for the inductive hypothesis to be ``strong enough to prove itself,'' we require an analogue of (\ref{eq:template}) both for the pair $(\yesd_\ell,\nod_\ell)$ and for another pair of distributions $(\yesd_\ell',\nod_\ell')$, and the inductive argument establishing the case $\ell=j+1$ from the case $\ell=j$ requires careful analysis of yet a third carefully constructed pair $(\yesd_\ell^\ast,\nod_\ell^\ast)$ of distributions.  See Section \ref{s:lower_bounds} for full details of the argument.)\hspace{0.04cm}\footnote{Notice that the argument we just sketched implies that $U_{d+1,N}$ is hard against depth-$d$ circuits. A more careful analysis at the end of the argument using the distributions $(\yesd_d^\ast,\nod_d^\ast)$ allows us to obtain the same lower bound for $U_{d,N}$, as stated in Theorem \ref{thm:lower}.}

% %%%%%%%%%%%%%%%%%%%%%%%%%%%%%%%%
% preliminaries
% %%%%%%%%%%%%%%%%%%%%%%%%%%%%%%%%

\medskip {\bf Notation and Organization.}
Recall that a \emph{restriction} $\rho$ of a function $f$ is an assignment fixing some of the input variables of $f$.  We write ``$f \uhr \rho$'' to denote $f$ restricted by $\rho$, a function over the rest of variables.  
We use boldface lower-case letters $\xx,\yy,$ etc. to denote string-valued random variables and boldface capital letters $\bX, \bY,$ etc. to denote real-valued random variables.

The rest of the paper is organized as follows.
We prove Theorems \ref{thm:lower} and \ref{thm:upper} in Sections \ref{s:lower_bounds}
  and \ref{s:upper_bounds}, respectively.
We then use Theorem \ref{thm:lower}  to prove Theorem \ref{t:strong_AG} in Section \ref{s:strong_AG}.

% %%%%%%%%%%%%%%%%%%%%%%%%%%%%%%%%
% lower_bounds
% %%%%%%%%%%%%%%%%%%%%%%%%%%%%%%%%

\section{The Lower Bound: Proof of Theorem \ref{thm:lower}.}\label{s:lower_bounds}

\def\xx{\bx}
\def\yy{\by}
\def\zz{{\bz}}

We prove Theorem \ref{thm:lower} in this section.
Throughout the section we use $d,n$ and $N$ to denote the three positive integers 
  in the statement of Theorem \ref{thm:lower} with $n\ge 2^{60d}$ and 
  $N \geq (2^{13} n)^d$.

This section is organized as follows.
In Sections \ref{sec:construct1} and \ref{sec:construct2}, we define inductively
  two pairs $(\yesd_\ell,\nod_\ell)$ 
  and $(\yesd_\ell',\nod_\ell')$ of distributions over strings
  $\xx\in \{0,1\}^{(\ell+{1})\times N_\ell}$ for $\ell$ from ${1}$ to $d$,
  where $N_\ell$ is specified later and satisfies $N_{1} < \cdots < N_d\le N$.~An important property of these distributions is that every $\xx$ drawn from 
  $\yesd_\ell$, $\nod_\ell$, $\yesd_\ell'$, $\nod_\ell'$ has $\SUM(\xx)$ equal to 
$$
2^{N_\ell},\quad 2^{N_\ell}-1,\quad 2^{N_\ell}-1\quad\text{and}\quad 2^{N_\ell}-(\ell+{1}),
$$ 
respectively. 
From the definition of $(\yesd_1,\nod_1)$ and $(\yesd'_1,\nod'_1)$, it is
  not too difficult to show that both~pairs are \emph{very hard} for monotone depth-$1$ majority circuits (Lemma \ref{lem:yes0no0}), i.e. 
  no majority gate with small weights can output 1 on strings drawn 
  from $\yesd_1$ with probability $p_1$ and at the same time 
  output $0$~on strings
  drawn from $\nod_1$ with probability $p_2$ if $p_1 + p_2$ is slightly larger than $1$ (and the same holds
  for $\yesd'_1$ and $\nod'_1$).
  
Then we prove our main technical lemma (Lemma \ref{l:inductive_step})
  in Section \ref{sec:induction}, which shows by 
  induction~that both pairs $(\yesd_\ell,\nod_\ell)$ and $(\yesd_\ell',\nod_\ell')$
  are hard in the same sense for ``small'' depth-{$\ell$} majority circuits 
  over $\{0,1\}^{(\ell+ {1})\times N_\ell}$ for every $\ell\in [d]$,
  with $(\yesd_1,\nod_1)$ and $(\yesd_1',\nod_1')$ serving as
  the base case.
Theorem \ref{thm:lower} for $U_{{d+1},N}$ (instead of $U_{{d},N}$ as stated)
  follows directely from $N_d\le N$ and the property that strings $\xx$ drawn from
  $\yesd_d$ and $\nod_d$ have $\SUM(\xx)$ equal to $2^{N_d}$ and $2^{N_d}-1$,
  respectively.
(Note that, although the second~pair   
  $(\yesd_d',\nod_d')$ is not needed in the proof of 
  Theorem \ref{thm:lower} once Lemma \ref{l:inductive_step} has been established, the intermediate 
  pairs $(\yesd_\ell',\nod_\ell')$ play a crucial role in the
  inductive definition of these distributions and the proof of Lemma \ref{l:inductive_step}.) 
  
  In order to extend the result to $U_{{d},N}$ (as stated
    in Theorem \ref{thm:lower}), we rely on  another auxiliary pair of distributions 
    $(\yesd_d^*,\nod_d^*)$ constructed during the proof, which is described in more detail in Section \ref{sec:construct2}.
We finally use Lemma \ref{l:inductive_step} to prove Theorem \ref{thm:lower} in Section \ref{sec:proof}.
  
\subsection{The Initial Two Pairs of Distributions.}\label{sec:construct1}

Let $d,n,N$ be positive integers in the statement of Theorem \ref{thm:lower}.
Let $\epsilon \eqdef 2^{-12d}$ and $N_1 \eqdef n \cdot (1/\epsilon)$. 
Given a string $z \in \{0,1\}^{2 \times N_1}$, 
  the \emph{$j$-th column of $z$} corresponds to a pair of positions $(1,j)$ and $(2,j)$, where $j \in [N_1]$. 

We now define two pairs of distributions $(\yesd_1,\nod_1)$ and
  $(\yesd_1',\nod_1')$ over $\{0,1\}^{2 \times N_1}$ and show that they are 
  hard for monotone depth-$1$ majority circuits of not-too-large size.
  We define the distributions via the following sampling processes.

\begin{flushleft}
\begin{itemize}
\item 
A string $\xx \sim \yesd_1$ is generated as follows. Let $\rr \sim [N_1]$ be uniformly random. We set both bits in the $\rr$-th column of $\xx$ to 1. For every $j > \rr$, we set both bits in the $j$-th column of $\xx$ to 0. For every $j < \rr$, we set the $j$-th column of $\xx$ to $(1,0)$ or $(0,1)$ independently and with equal probability.  For example, writing an $x \in \supp(\yesd_1)$ as a matrix, it would look like
\begin{center}
$x = $
\begin{tabular}{cccccc|c|ccccc}
1 & 0 & 0 & 1 & $\cdots$ & 0 \hspace{0.012cm}&\hspace{0.012cm} 1 \hspace{0.012cm}& \hspace{0.01cm} 0 & 0 & $\cdots$ & 0 & 0\\
0 & 1 & 1 & 0 & $\cdots$ & 1 \hspace{0.012cm}&\hspace{0.012cm} 1 \hspace{0.012cm}& \hspace{0.01cm} 0 & 0 & $\cdots$ & 0 & 0\\
\end{tabular}
\quad and we have \quad $\SUM(x)=2^{N_1}$.
\end{center}

\item A string $\yy \sim \nod_1$ is generated by setting its $j$-th column to $(1,0)$ or $(0,1)$ independently and with equal probability for each $j \in [\rr]$.  So a string $y\in \supp(\nod_1)$ would look like
\begin{center}
$y = $
\begin{tabular}{cccccccccccc}
0 & 1 & 0 & 0 & $\cdots$ & 1 & 1 & 0 & 1 & $\cdots$ & 1 & 0\\
1 & 0 & 1 & 1 & $\cdots$ & 0 & 0 & 1 & 0 & $\cdots$ & 0 & 1\\
\end{tabular}\quad and we have \quad $\SUM(y)=2^{N_1} - 1$.
\end{center}

\item $\yesd_1'$ is the same as $\nod_1$. In particular, each $x \in \supp(\yesd_1')$ has $\SUM(x)=2^{N_1} - 1$.

\item Finally, a string $\yy \sim \nod_1'$ is obtained as follows. First, sample a random $\xx \sim \yesd_1$. Then let $\yy$ be the string obtained by negating each bit of $\xx$. So a string $y \in \supp(\nod_1')$ looks like
\begin{center}
$y = $
\begin{tabular}{cccccc|c|ccccc}
1 & 1 & 0 & 1 & $\cdots$ & 0 \hspace{0.012cm}&\hspace{0.012cm} 0 \hspace{0.012cm}&\hspace{0.012cm} 1 & 1 & $\cdots$ & 1 & 1\\
0 & 0 & 1 & 0 & $\cdots$ & 1 \hspace{0.012cm}&\hspace{0.012cm} 0 \hspace{0.012cm}&\hspace{0.012cm} 1 & 1 & $\cdots$ & 1 & 1\\
\end{tabular}
\quad and\quad $\SUM(y)= 2^{N_1}-2$.
\end{center}
\end{itemize}\end{flushleft}

Recall a monotone depth-$1$ majority circuit of size $s$ is simply a monotone weighted majority gate with total weight at most $s$.
We show below that both pairs $(\yesd_1,\nod_1)$ and
  $(\yesd_1',\nod_1')$ defined above are hard for a monotone depth-1 circuit
  (to be correct on both $\yesd_1$ and $\nod_1$, 
  or on both $\yesd_1'$ and $\nod_1'$, with nontrivial probability) {unless the total weight $s$ is large.}

\begin{lemma} \label{lem:yes0no0}
For any depth-$1$ {monotone} majority circuit $F$ over $\{0,1\}^{2 \times N_1}$ of size at most $2^{n -1}$, 
\begin{align}&\Pr_{\xx \sim \yesd_1}\big[F(\xx)=1\big] + 
\Pr_{\yy \sim \nod_1}\big[F(\yy)=0\big]  \leq 1 + \epsilon, 
  \label{eq:oneplustauinit1}\quad\text{and}\\[0.7ex]
&\Pr_{\xx \sim \yesd'_1}\big[F(\xx)=1\big] + 
\Pr_{\yy \sim \nod_1'}\big[F(\yy)=0\big]  \leq 1 + \epsilon\label{eq:oneplustauinit2}.
\end{align}
\end{lemma}

\begin{proof}
We present the proof of the first inequality on $(\yesd_1,\nod_1)$. An entirely similar argument establishes the bound for $(\yesd'_1, \nod'_1)$. 

Consider an auxiliary distribution $\mathcal{D}$ (essentially a coupling of $\yesd_1$ and $\nod_1$) supported over $\{0,1\}^{ {2 \times}N_1} \times \{0,1\}^{ {2 \times} N_1} \times [N_1]$, and defined in the following way. A draw $(\xx,\yy, \rr) \sim \mathcal{D}$ is obtained by selecting a uniformly random $\rr \sim [N_1]$, a string $\yy \sim \nod_1$, and by letting $\xx = \xx(\yy, \rr) \in \{0,1\}^{ {2 \times}N_1}$ be the string obtained by replacing the $\rr$-th column of $\yy$ with $(1,1)$, and by setting the $j$-th column of $\yy$ to $(0,0)$ for every $j > \rr$. Observe that the marginal distributions $\mathcal{D}_{\xx}$ and $\mathcal{D}_{\yy}$ are identical to $\yesd_1$ and $\nod_1$, respectively. Consequently,
\begin{align*}
\text{LHS of Equation}~(\ref{eq:oneplustauinit1}) & = \Pr_{(\xx,\yy,\rr) \sim \mathcal{D}} \big [ F(\xx) = 1  \big ] + \Pr_{(\xx,\yy,\rr) \sim \mathcal{D}} \big [ F(\yy) = 0  \big ] \nonumber \\[0.4ex] 
& = \Pr \big [F(\xx) = 1~\text{or}~F(\yy) = 0 \big] +  \Pr \big [F(\xx) = 1~\text{and}~F(\yy) = 0\big] \nonumber \\[0.4ex]
& \leq 1 + \Pr \big [F(\xx) = 1~\text{and}~F(\yy) = 0\big]. \nonumber
\end{align*}
Hence to prove the lemma, it is enough to show that
\begin{equation}\label{eq:eps_bound}
q \eqdef \Pr_{(\xx,\yy,\rr) \sim \mathcal{D}} \big [F(\xx) = 1~\text{and}~F(\yy) = 0\big]  \leq \epsilon.
\end{equation}

For every $r \in [N_1]$, let $\bY_r$ be an indicator random variable defined on $\mathcal{D}$ that is $1$ whenever
$$
w_r(\yy) \;>\; \sum_{\ell \,>\, r} w_\ell(\yy),
$$
where $w_j(y) \eqdef w_{1,r} \cdot y_{1,r} + w_{2,r} \cdot y_{2,r}$, and $w_{i,j}$ is the weight corresponding to the input variable of $F$ at position $(i,j)$. Informally, $\bY_r = 1$ if and only if the weight of $\yy$ with respect to $F$ at the $r$-th column is strictly larger than the sum of the weights collected from all succeeding columns.

We will employ the following claim to establish Equation (\ref{eq:eps_bound}).
\begin{claim}\label{claim:Y_r_inequality}
For every $j \in [N_1]$, we have 
$$
q_j \eqdef \Pr_{(\xx,\yy,\rr) \sim \mathcal{D}} \big [F(\xx) = 1~\text{\emph{and}}~F(\yy) = 0\hspace{0.05cm}\big|\hspace{0.05cm}\rr = j\big]  \leq  \Pr_{\mathcal{D}} \big [ \bY_j = 1 \big ].
$$
\end{claim}
\begin{proof}
We consider first the case where $j = 1$. 
The conditions of $F(\xx) = 1$ and $\rr = 1$ imply that
$w_{1,1} + w_{2,1}\ge t,$
where $t$ is the threshold of $F$. Furthermore, because $F(\yy) = 0$ it must be the case that
$\sum_{r = 1}^{N_1} w_r(\yy) < t.$
These inequalities give us
\begin{equation}\label{eq:x_and_y_combined}
w_{1,1} + w_{2,1} - w_1(\yy)  >  \sum_{r \,>\, 1} w_r(\yy).
\end{equation}
Let $\widetilde{\yy}$ be the string obtained from $\yy$ by flipping the two bits in the first column of $\yy$. Equation (\ref{eq:x_and_y_combined}) is then equivalent to $w_1(\widetilde{\yy})>\sum_{r \,>\, 1} w_r(\widetilde{\yy}).$ Therefore,
\begin{align*}
q_1& \leq \Pr_{(\xx,\yy,\rr) \sim \mathcal{D}} \Big [w_1(\widetilde{\yy}) > \sum_{r > 1} w_r(\widetilde{\yy})
\hspace{0.05cm}\big|\hspace{0.05cm}\rr = 1\Big] \nonumber\\
& = \Pr_{(\xx,\yy,\rr) \sim \mathcal{D}} \Big [w_1(\yy) > \sum_{r > 1} w_r(\yy)
\hspace{0.05cm}\big|\hspace{0.05cm}\rr = 1\Big] \nonumber \\
& = \Pr_{(\xx,\yy,\rr) \sim \mathcal{D}} \big [ \bY_1 = 1 \big], \nonumber
\end{align*}
where the last two equations use the independence of $\yy$ and $\rr$ as well as the fact that $\widetilde{\yy}$ and $\yy$ are identically distributed.

For $j > 1$ the result can be proved similarly by writing $q_j$ as a conditional expectation over the outcome of the first $j - 1$ columns of $\yy$, then adapting the argument above in the natural way. 
\end{proof}

Claim \ref{claim:Y_r_inequality} and the definitions of probabilities $q$ and $q_j$ imply that
$$
N_1 \cdot q \;=\; \sum_{j = 1}^{N_1} q_j \;\leq\; \sum_{j=1}^{N_1} \Pr_{\mathcal{D}} \big [ \bY_j = 1 \big ]\; = \;\E_{\mathcal{D}} \left[ \sum_{j = 1}^{N_1} \bY_j \right].
$$
In particular, there is a string $y^\ast\in \supp(\nod_1)$ and a set $S \subseteq [N_1]$ with $|S| \geq N_1 \cdot q$ such that
\begin{equation}\label{hehe2}
w_r(y^\ast)  > \sum_{\ell \,>\, r} w_\ell(y^\ast),
\end{equation} 
for each $r \in S$.
Recall that    the weight associated to each variable in $F$ is a non-negative integer,
   and that the total weight of $F$ is at least
   $\sum_{r\ge 1} w_r(y^\ast)$. It follows directly from 
   (\ref{hehe2}) that $F$ must have total weight at least $\smash{2^{|S| - 1}}$. However, by assumption $F$ has total weight at most $2^{n - 1}$. Altogether, we get from these inequalities and $N_1 = n \cdot (1/\epsilon)$  that $q \leq \epsilon$, which completes the proof.
\end{proof}  

\subsection{A Sequence of Pairs of Pairs of Distributions.} \label{sec:construct2} 
  
Next, suppose that we have defined pairs of distributions $(\yesd_{\ell-1},\nod_{\ell-1})$ and
  $(\yesd_{\ell-1}',\nod_{\ell-1}')$ over $\{0,1\}^{\ell \times N_{\ell - 1}}$ 
  for some ${2 \leq \ell \leq d}$, 
where a string $\xx$ drawn from $\yesd_{\ell-1}$, $\nod_{\ell-1}$, $\yesd_{\ell-1}'$
  and $\nod_{\ell-1}'$
  has $\SUM(\xx)$ equal to 
\begin{equation}\label{eq:sum0}
2^{N_{\ell-1}},\quad 2^{N_{\ell-1}}-1,\quad 2^{N_{\ell-1}}-1\quad \text{and}\quad
2^{N_{\ell-1}}- {((\ell - 1) + 1)},
\end{equation}
respectively.  (Note that the pairs $(\yesd_{1},\nod_{1})$ and
  $(\yesd_{1}',\nod_{1}')$ have this property.)
Our aim is to inductively define $(\yesd_\ell,\nod_\ell)$ and 
  $(\yesd_\ell',\nod_\ell')$ over $\{0,1\}^{(\ell+ {1}) \times N_\ell}$, where 
  $$
N_\ell \eqdef n\cdot N_{\ell-1}+1
  \leq   {2^\ell \cdot n^{\ell - 1} \cdot N_1} =(2n)^\ell \cdot  {2^{12}}^d \leq  {(2^{13}n)^d} \leq N, \quad \text{for~}\ell \in \{2, \ldots, d\},
$$
and a string $\xx$ drawn from $\yesd_\ell$,
  $\nod_\ell$, $\yesd_\ell'$ and $\nod_\ell'$ has $\SUM(\xx)$ 
  equal to 
\begin{equation}\label{eq:sum2}
2^{N_\ell},\quad 2^{N_\ell}-1,\quad 2^{N_\ell}-1\quad \text{and}\quad 2^{N_\ell}-(\ell+ {1}),
\end{equation}
respectively.
To this end we start by defining a pair of distributions $(\yesd_\ell^*,\nod_\ell^*)$ over
  $\{0,1\}^{ {\ell} \times N_\ell^\ast}$ (note that the number of rows for these  distributions, ${\ell}$, is exactly the same as 
  for the distributions $\yesd_{\ell-1},\nod_{\ell-1}$ and
  $\yesd_{\ell-1}',\nod_{\ell-1}'$), with $$N_\ell^*\eqdef n\cdot N_{\ell-1}=N_\ell-1.$$  
  
\begin{figure}[t!]
\captionsetup{justification=centering}
\begin{center}
\begin{picture}(400,190)(20, 40)

% Top picture:
\put(-1,159){$\underbrace{\mbox{~~~~~~~~~~~~~~~~~~~~~~~~~~~~~~~~~~~~~~~~~~~~~~~~~~~~~~~~~~~~~~~~~~~~~~~~~~~~~~~~~~~~~~~~~~~~~~~~~~~~~~~~~~~~~~~~~~~~~~~~~~~~}}$}
\put(223,143){{\small $\xx \sim \yesd^\ast_\ell$}}
\linethickness{1pt}
\put(0,200){\line(1,0){450}}
\put(0,160){\line(1,0){450}}
\put(0,200){\line(0,-1){40}}
\put(450,200){\line(0,-1){40}}

\put(5,215){{\small section}}
\put(5,205){{\small 1}}
\put(10,187){{\small $\yesd'_{\ell-1}$}}
\put(17,178){{\small or}}
\put(10,167){{\small $\nod_{\ell-1}$}}

\linethickness{0.5pt}
\put(50,200){\line(0,-1){40}}
\put(150,200){\line(0,-1){40}}
\put(200,200){\line(0,-1){40}}
\put(250,200){\line(0,-1){40}}
\put(300,200){\line(0,-1){40}}
\put(400,200){\line(0,-1){40}}

\put(56,178){{\small $\cdots \cdots \cdots \cdots \cdots \cdots \cdot \cdot$}}

\put(155,215){{\small section}}
\put(155,205){{\small $\bt-1$}}
\put(160,187){{\small $\yesd'_{\ell-1}$}}
\put(167,178){{\small or}}
\put(160,167){{\small $\nod_{\ell-1}$}}

\put(205,215){{\small section}}
\put(205,205){{\small $\bt$}}
\put(210,178){{\small $\yesd_{\ell-1}$}}

\put(255,215){{\small section}}
\put(255,205){{\small $\bt+1$}}

\put(252,190){{\small $0 \hskip-1pt\cdots \cdots \cdot \cdot \hskip1pt 0$}}
\put(253.2,177){{\small $\vdots$}} 
\put(253.2,171){{\small $\cdot$}}
\put(252,163){{\small $0 \hskip-1pt\cdots \cdots \cdot \cdot \hskip1pt 0$}}
\put(295,171){{\small $\cdot$}}
\put(295,177){{\small $\vdots$}} 

\put(306,178){{\small $\cdots \cdots \cdots \cdots \cdots \cdots \cdot \cdot$}}

\put(405,215){{\small section}}
\put(405,205){{\small $n$}}
\put(402,190){{\small $0 \hskip-1pt\cdots \cdots \cdot \cdot \hskip1pt 0$}}
\put(403.2,177){{\small $\vdots$}} 
\put(403.2,171){{\small $\cdot$}}
\put(402,163){{\small $0 \hskip-1pt\cdots \cdots \cdot \cdot \hskip1pt 0$}}
\put(445,171){{\small $\cdot$}}
\put(445,177){{\small $\vdots$}}

% Bottom picture:
\put(-1,59){$\underbrace{\mbox{~~~~~~~~~~~~~~~~~~~~~~~~~~~~~~~~~~~~~~~~~~~~~~~~~~~~~~~~~~~~~~~~~~~~~~~~~~~~~~~~~~~~~~~~~~~~~~~~~~~~~~~~~~~~~~~~~~~~~~~~~~~~}}$}
\put(223,43){{\small $\xx \sim \nod^\ast_\ell$}}
\linethickness{1pt}
\put(0,100){\line(1,0){450}}
\put(0,60){\line(1,0){450}}
\put(0,100){\line(0,-1){40}}
\put(450,100){\line(0,-1){40}}

\put(5,115){{\small section}}
\put(5,105){{\small 1}}
\put(10,87){{\small $\yesd'_{\ell-1}$}}
\put(17,78){{\small or}}
\put(10,67){{\small $\nod_{\ell-1}$}}

\linethickness{0.5pt}
\put(50,100){\line(0,-1){40}}
\put(150,100){\line(0,-1){40}}
\put(200,100){\line(0,-1){40}}
\put(250,100){\line(0,-1){40}}
\put(300,100){\line(0,-1){40}}
\put(400,100){\line(0,-1){40}}

\put(56,78){{\small $\cdots \cdots \cdots \cdots \cdots \cdots \cdot \cdot$}}

\put(155,115){{\small section}}
\put(155,105){{\small $\bt-1$}}
\put(160,87){{\small $\yesd'_{\ell-1}$}}
\put(167,78){{\small or}}
\put(160,67){{\small $\nod_{\ell-1}$}}

\put(205,115){{\small section}}
\put(205,105){{\small $\bt$}}
\put(210,78){{\small $\nod'_{\ell-1}$}}

\put(255,115){{\small section}}
\put(255,105){{\small $\bt+1$}}

\put(252,90){{\small $1 \cdots \cdots \cdot \cdot 1$}}
\put(253.2,77){{\small $\vdots$}} 
\put(253.2,71){{\small $\cdot$}}
\put(252,63){{\small $1 \cdots \cdots \cdot \cdot 1$}}
\put(295,71){{\small $\cdot$}}
\put(295,77){{\small $\vdots$}} 

\put(306,78){{\small $\cdots \cdots \cdots \cdots \cdots \cdots \cdot \cdot$}}

\put(405,115){{\small section}}
\put(405,105){{\small $n$}}
\put(402,90){{\small $1 \cdots \cdots \cdot \cdot 1$}}
\put(403.2,77){{\small $\vdots$}} 
\put(403.2,71){{\small $\cdot$}}
\put(402,63){{\small $1 \cdots \cdots \cdot \cdot 1$}}
\put(445,71){{\small $\cdot$}}
\put(445,77){{\small $\vdots$}} 
\end{picture}\vspace{0.25cm}
\caption{Illustrations of how the $\yesd_\ell^\ast$ and $\nod_\ell^\ast$ distributions\vspace{0.03cm}\\ 
are defined from the {$\yesd_{\ell-1}',\nod_{\ell - 1}',\yesd_{\ell - 1}$ and $\nod_{\ell - 1}$} distributions.\vspace{-0.4cm}}
\label{illustration0}
\end{center}
\end{figure}
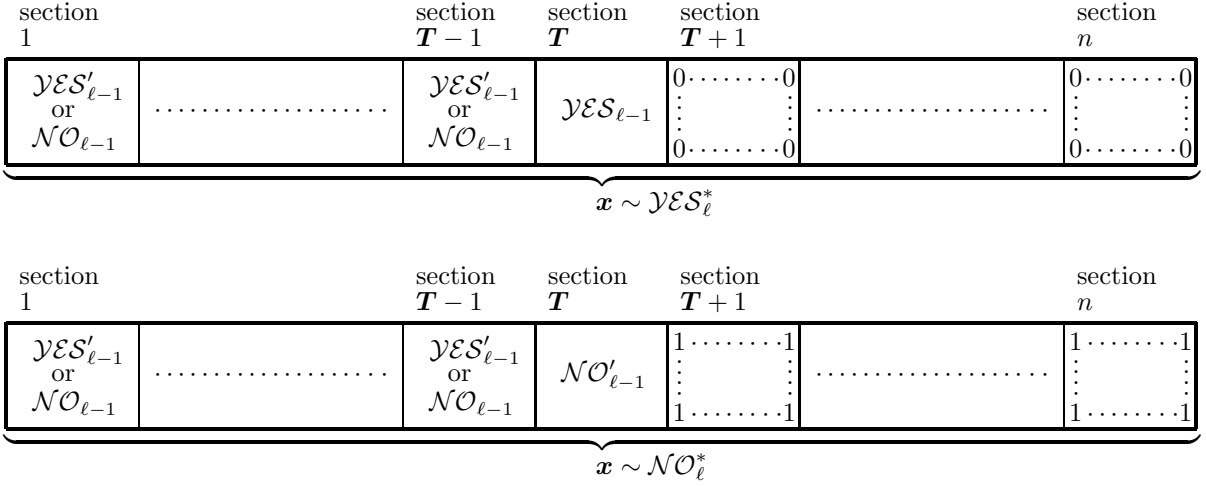

To define $(\yesd_\ell^*,\nod_\ell^*)$, we partition the $N_\ell^*$ columns into $n$ \emph{sections}, each with $N_{\ell-1}$ columns (and {$\ell$} rows). 
(So the first section consists of all $x_{i,j}$ with $j\in [N_{\ell-1}]$, the second 
  section consists of all $x_{i,j}$ with $j\in [N_{\ell-1}+1,2N_{\ell-1}]$, and so forth.)
A draw of a string from  $\yesd_\ell^*$ is obtained as follows:  first we draw an integer $\bt$ uniformly from $[n]$, and then
\begin{flushleft}\begin{enumerate}
\item [$(a)$] For each $i<\bt$, we independently set the $i$-th section to be a string drawn
  from $\nod_{\ell-1}$ with probability $1/2$ or a string drawn from $\yesd_{\ell-1}'$ with probability $1/2$.\vspace{-0.08cm}
\item [$(b)$] For each $i>\bt$, we set the $i$-th section to be all $0$.\vspace{-0.08cm}
\item [$(c)$] For the $\bt$-th section, we set it to be a string drawn from $\yesd_{\ell-1}$.
\end{enumerate}\end{flushleft}
See Figure \ref{illustration0} for an illustration.
A draw of a string from $\nod_\ell^*$ is obtained in a similar fashion.
First we draw $\bt$ from $[n]$ uniformly at random, and then
\begin{flushleft}\begin{enumerate}
\item [$(a')$] For each $i<\bt$, we independently  set the $i$-th section to be a string drawn
  from $\nod_{\ell-1}$ with probability $1/2$ or a string drawn from $\yesd_{\ell-1}'$ with probability $1/2$.
  (Note that this is the same as step $(a)$ above in the definition of $\yesd_\ell^*$.) \vspace{-0.08cm}
\item [$(b')$]  For each $i>\bt$, we set the $i$-th section to be all $1$ (this is different from $(b)$ above).\vspace{-0.08cm}
\item [$(c')$] For the $\bt$-th section, we set it to be a string drawn from $\nod_{\ell-1}'$ (this is different from $(c)$).
\end{enumerate}\end{flushleft}
Again see Figure \ref{illustration0} for an illustration.
Given (\ref{eq:sum0}), we see that a string $\xx$
  drawn from $\yesd_\ell^*$ (or from $\nod_\ell^*$)
  has $\SUM(\xx)$ equal to $2^{N_\ell^*}$ (respectively, equal to $2^{N_\ell^*}- {\ell}$).

With the definitions of $\yesd_\ell^*$ and $\nod_\ell^*$ in hand, we now use them to define
  $(\yesd_\ell,\nod_\ell)$ and $(\yesd_\ell',\nod_\ell')$ so that every 
  string $\xx$ drawn from these distributions should
  have $\SUM(\xx)$ equal to the values given in (\ref{eq:sum2}).
Recall that $N_\ell=N^*_\ell+1$.

A string $\xx=(x_{i,j})\in \{0,1\}^{(\ell+ {1}) \times N_\ell}$ drawn from
   $\yesd_\ell'$ is obtained as follows.
First we draw a string $\zz$ from $\yesd_\ell^*$ and put it in 
  columns $\{2,\ldots,N_\ell\}$ and rows $\{1,\dots, {\ell}\}$ of $\xx$, i.e.,
  $x_{i,j}=z_{i,j-1}$ for all $i\in [ {\ell}]$ and $j\in \{2,\ldots,N_\ell\}$.
  For the remaining positions (in the first column and the last row), we set
  $x_{i,1}=0$ for all $i\in [ {\ell+1}]$ and $x_{\ell+ {1},j}=1$
  for all $j>1$.
The other distribution $\nod_\ell'$ is defined similarly, 
  except that we draw the string $\zz$ from $\nod_\ell^*$ instead of from $\yesd_\ell^*$.
The definition of $\yesd_\ell'$ and $\nod_\ell'$ is illustrated in Figure \ref{illustration}.

\begin{figure}[t]
\captionsetup{justification=centering}
\begin{center}
\begin{picture}(400,190)(-40, 40)
% Top left box:
\put(68,212){{\small $\xx \sim \yesd'_\ell$}}
\put(-1,202){$\overbrace{\mbox{~~~~~~~~~~~~~~~~~~~~~~~~~~~~~~~~~~~~~~~}}$}
\linethickness{1pt}
\put(0,200){\line(1,0){139}}
\put(0,153){\line(1,0){139}}
\put(0,200){\line(0,-1){47}}
\put(139,200){\line(0,-1){47}}
\linethickness{0.5pt}
\put(10,200){\line(0,-1){47}}
 \put(0,165){\line(1,0){139}}
\linethickness{1pt}
\put(3,191){{\small $0$}}
\put(4.2,178){{\small $\vdots$}} 
\put(3,168){{\small $0$}}
\put(55,180){{\small $\zz \sim \yesd_\ell^\ast$}}
\put(3,156){{\small $0$~$1 \cdots \cdots \cdots \cdots \cdots \cdots \cdots \cdots \cdot \cdot 1$}}

% Top right box:
\put(268,212){{\small $\xx \sim \nod'_\ell$}}
\put(199,202){$\overbrace{\mbox{~~~~~~~~~~~~~~~~~~~~~~~~~~~~~~~~~~~~~~~}}$}
\linethickness{1pt}
\put(200,200){\line(1,0){139}}
\put(200,153){\line(1,0){139}}
\put(200,200){\line(0,-1){47}}
\put(339,200){\line(0,-1){47}}
\linethickness{0.5pt}
\put(210,200){\line(0,-1){47}}
 \put(200,165){\line(1,0){139}}
\linethickness{1pt}
\put(203,191){{\small $0$}}
\put(204.2,178){{\small $\vdots$}} 
\put(203,168){{\small $0$}}
\put(255,180){{\small $\zz \sim \nod_\ell^\ast$}}
\put(203,156){{\small $0$~$1 \cdots \cdots \cdots \cdots \cdots \cdots \cdots \cdots \cdot \cdot 1$}}

% Bottom left box:
\put(68,112){{\small $\xx \sim \yesd_\ell$}}
\put(-1,102){$\overbrace{\mbox{~~~~~~~~~~~~~~~~~~~~~~~~~~~~~~~~~~~~~~~}}$}
\linethickness{1pt}
\put(0,100){\line(1,0){139}}
\put(0,53){\line(1,0){139}}
\put(0,100){\line(0,-1){47}}
\put(139,100){\line(0,-1){47}}
\linethickness{0.5pt}
\put(10,100){\line(0,-1){47}}
 \put(0,65){\line(1,0){139}}
\linethickness{1pt}
\put(3,91){{\small $0$}}
\put(4.2,78){{\small $\vdots$}} 
\put(3,68){{\small $0$}}
\put(55,80){{\small $\zz \sim \yesd_\ell^\ast$}}
\put(3,56){{\small $1$~\hskip1.3pt$0\cdots \cdots \cdots \cdots \cdots \cdots \cdots \cdots \cdot  0$}}

% Bottom right box:
\put(268,112){{\small $\xx \sim \nod_\ell$}}
\put(199,102){$\overbrace{\mbox{~~~~~~~~~~~~~~~~~~~~~~~~~~~~~~~~~~~~~~~}}$}
\linethickness{1pt}
\put(200,100){\line(1,0){139}}
\put(200,53){\line(1,0){139}}
\put(200,100){\line(0,-1){47}}
\put(339,100){\line(0,-1){47}}
\linethickness{0.5pt}
\put(210,100){\line(0,-1){47}}
 \put(200,65){\line(1,0){139}}
\linethickness{1pt}
\put(203,91){{\small $0$}}
\put(204.2,78){{\small $\vdots$}} 
\put(203,68){{\small $0$}}
\put(255,80){{\small $\zz \sim \nod_\ell^\ast$}}
\put(203,56){{\small $1$~~(binary represent. of $\ell - 1$)}}

\end{picture}\vspace{0.3cm}
\caption{Illustrations of how the $\yesd_\ell',\nod_\ell',\yesd_\ell$ and $\nod_\ell$ distributions\vspace{0.03cm}\\ are defined from the $\yesd_\ell^\ast$ and $\nod_\ell^\ast$ distributions.}\vspace{-0.3cm}
\label{illustration}
\end{center}
\end{figure}
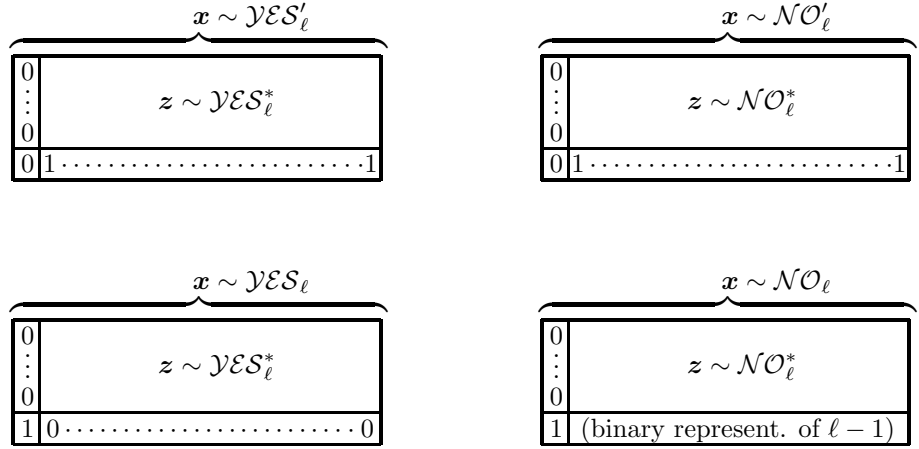

For the other pair $(\yesd_\ell,\nod_\ell)$,
  a string $\xx$ drawn from $\yesd_\ell$ is obtained as follows.  As before, 
  we first draw a string $\zz$ from $\yesd_\ell^*$ and put it in
  columns $\{2,\ldots,N_\ell\}$ and rows $\{1,\dots, {\ell}\}$ of $\xx$.
Then we set $x_{\ell+ {1},1}=1$ and all other variables on the 
  first row and last column of $\xx$ to be $0$.
For the other distribution $\nod_\ell$, we similarly draw $\zz$ from $\nod_\ell^*$ and 
  put it in columns $\{2,\ldots,N_\ell\}$ and rows $\{1,\dots, {\ell}\}$ of $\xx$.
We set~$x_{\ell+ {1},1}=1$ and all other variables on the first column to be $0$.
We set the last row, i.e., $x_{\ell+ {1},j}$ with $j\in \{2,\ldots,N_\ell\}$,
  to be the binary representation of $\ell-1$. 
(This is well defined since $N^\ast_\ell \geq n\ge 2^{{60}d}\gg \log d \ge \log \ell$.)
As before, see Figure \ref{illustration} for an illustration of the definition of $\yesd_\ell'$ and $\nod_\ell'$.

We record the following two useful facts about $N_\ell$ and the distributions:

\begin{fact}
$N_d\le N$.
\end{fact}

\begin{fact}
For each $\ell \in [d]$, a string $\xx$ drawn from $\yesd_\ell$, $\nod_\ell$,
  $\yesd_\ell'$, $\nod_\ell'$
  has $\SUM(\xx)$ equal to $$2^{N_\ell},\quad 2^{N_\ell}-1,\quad 
  2^{N_\ell}-1\quad \text{and} \quad 2^{N_\ell}-(\ell+ {1}).$$ \end{fact}

An important property of the $(\yesd_\ell,\nod_\ell)$ pair and of the $(\yesd_\ell',\nod_\ell')$ pair --- which in fact
motivated the above definitions of these distributions in terms of $\yesd_{\ell}^\ast$ and $\nod_{\ell}^\ast$ --- is that 
they are \emph{at least as hard} to distinguish as $(\yesd_\ell^*,\nod_\ell^*)$ for monotone majority circuits.

This is made formal in the following two lemmas.

\begin{lemma}\label{lem:trivial1}
Given any monotone majority circuit $F$ over $\{0,1\}^{(\ell+1) \times N_\ell}$, 
  there is a monotone majority circuit $F^\ast$ over $\{0,1\}^{\ell \times N_\ell^\ast}$ of the same size and depth as $F$ such that
$$
\Pr_{\xx\in \yesd_\ell'}\big[F(\xx)=1\big]+\Pr_{\yy\in \nod_\ell'}\big[F(\yy)=0\big]
= \Pr_{\xx\in \yesd_\ell^*}\big[F^\ast(\xx)=1\big]+\Pr_{\yy\in \nod_\ell^*}\big[F^\ast(\yy)=0\big].
$$
\end{lemma}
\begin{proof}
Given $F$, we hard-wire the variables in the first column to be $0$ and the
  rest of the variables in the last row to be $1$.
Let $F^\ast$ denote the new monotone majority circuit obtained from $F$
  of the same size and depth.
The definition of $\yesd_\ell'$ and $\nod_\ell'$ from $\yesd_\ell^*$ and $\nod_\ell^*$ 
  implies that
\begin{align*}
\Pr_{\xx\in \yesd_\ell'}\big[F(\xx)=1\big]&=\Pr_{\xx\in \yesd_\ell^*}\big[F^\ast(\xx)=1\big] \ \ \ \text{and}
\\[0.4ex]
 \Pr_{\yy\in \nod_\ell'}\big[F(\yy)=0\big]
&= \Pr_{\yy\in \nod_\ell^*}\big[F^\ast(\yy)=0\big].
\end{align*}
The lemma then follows.
\end{proof}

\begin{lemma}\label{lem:trivial2}
Given any  monotone majority circuit $F$ over $\{0,1\}^{(\ell+1) 
  \times N_\ell}$,
  there~is a  monotone~majority circuit $F^\ast$ over $\{0,1\}^{\ell \times N_\ell^\ast}$ of the same size and depth as $F$ such that
$$
\Pr_{\xx\in \yesd_\ell }\big[F(\xx)=1\big]+\Pr_{\yy\in \nod_\ell }\big[F(\yy)=0\big]
\le \Pr_{\xx\in \yesd_\ell^*}\big[F^\ast(\xx)=1\big]+\Pr_{\yy\in \nod_\ell^*}\big[F^\ast(\yy)=0\big].
$$
\end{lemma}
\begin{proof}
Given $F$, we hard-wire $x_{\ell+ {1},1}$ to be $1$ and the rest of the variables
  in the first column and the last row to be $0$.
Let $F^\ast$ denote the resulting  monotone majority circuit obtained from 
  $F$ of the same size and depth.
The definition of $\yesd_\ell$ and $\nod_\ell$ from $\yesd_\ell^*$ and $\nod_\ell^*$ implies that
\begin{align*}
\Pr_{\xx\in \yesd_\ell}\big[F(\xx)=1\big]&=\Pr_{\xx\in \yesd_\ell^*}\big[F^\ast(\xx)=1\big] \ \ \ \text{and}
\\[0.4ex]
 \Pr_{\yy\in \nod_\ell}\big[F(\yy)=0\big]
&\le \Pr_{\yy\in \nod_\ell^*}\big[F^\ast(\yy)=0\big],
\end{align*}
where the inequality follows from the monotonicity of $F$.
The lemma then follows.
\end{proof}

\subsection{The Key Induction Lemma.}\label{sec:induction}

Given distributions defined in Sections \ref{sec:construct1} and \ref{sec:construct2},
  we prove the following key technical lemma.
  
Recall that $\epsilon=2^{-12d}$. Below we let $M=2^{\epsilon^5 n}$.

\begin{lemma}\label{l:inductive_step}
Let $\ell\in \{2, \ldots, d\}$.
Suppose that any depth-$(\ell-1)$ monotone majority circuit $F$ over $\{0,1\}^{{\ell} \times N_{\ell-1}}$ 
  of size at most $M$ satisfies
\begin{align}
&\Pr_{\bx\sim \yesd_{\ell-1}} \big[ F(\bx)=1\big]+\Pr_{\by\sim \nod_{\ell-1}}\big[F(\by)=0\big]\le 1+7^{\ell- {2}}\epsilon\ \ \ \text{and}\nonumber\\[0.5ex] &
\Pr_{\bx\sim \yesd_{\ell-1}'} \big[ F(\bx)=1\big]+ \Pr_{\by\sim \nod_{\ell-1}'}\big[F(\by)=0\big]\le 1+7^{\ell- {2}}\epsilon.
\label{eq:condcond}
\end{align}
Then any depth-$\ell$ monotone majority circuit $F^\ast$ over $\{0,1\}^{ {\ell} \times N^\ast_\ell}$
  of size at most $M$ satisfies
$$
\Pr_{\bx\sim \yesd_{\ell}^*} \big[ F^\ast(\bx)=1\big]+\Pr_{\by\sim \nod_{\ell}^*}\big[F^\ast(\by)=0\big]\le 1+7^{\ell - {1}} \epsilon.
$$
\end{lemma}
\begin{proof}
Recall that strings drawn from $\yesd_\ell^*$ and $\nod_\ell^*$ consist of $n$ sections.
For convenience, we refer to strings in $\{0,1\}^{{\ell} \times N_{\ell-1}}$ as  
  \emph{section strings}.

We begin by defining some useful distributions $\calD_{1},\dots,\calD_{n}$ over concatenations of section strings where $\calD_t$ is supported on concatenations of ${t-1}$ section strings.
 First,
let $\calD$ denote the following  distribution over section strings:
  $\xx\sim \calD$ is drawn from
  $\nod_{\ell-1}$ with probability $1/2$ and is drawn from $\yesd_{\ell-1}'$ with probability $1/2$. 
For each $t \in [n]$, we use $\calD_t$ to  denote the distribution of 
the  concatenation of ${t-1}$ section strings, each drawn from $\calD$
   independently.
(So $\calD_t$ is a distribution over $\{0,1\}^{{\ell} \times (t-1)N_{\ell-1}}$.)
Note that in the special case when $t=1$, $\calD_{1}$ is supported on the empty string only.
Note also that for $t \in [n]$, $\calD_t$ is generated precisely according to $(a)$ or $(a')$ from Section \ref{sec:construct2}  (recall that $(a)$ and $(a')$ are the same).

As in the statement of Lemma \ref{l:inductive_step}, let $F^\ast$ be a depth-$\ell$ monotone majority circuit on $\{0,1\}^{{\ell} \times N_\ell^*}$
  of size at most $M$.
We say a string $z\in \supp(\calD_t)$ for some $t \in [n]$ is \emph{good}
   with respect to $F^\ast$ if 
$$
\Pr_{\bx\sim \yesd_{\ell-1}}\big[F^\ast(z\circ \bx\circ \00)=1\big] +\Pr_{\by\sim \nod_{\ell-1}'}
  \big[F^\ast(z\circ \by\circ \11)=0\big]\ge 1+6\hspace{0.02cm}\delta,
$$
where we write $\00$ and $\11$ to denote the all-0 and all-1 strings
  in $\{0,1\}^{ {\ell} \times (n-t)N_\ell}$, and 
$  \delta \eqdef 7^{\ell- {2}}\epsilon.$  
  
Now we fix a $t \in [n]$ and fix a good string $z\in \supp(\calD_t)$.
Let $\rho_{z}$ be the restriction that fixes the first $t-1$ sections
  of variables of $F^\ast$ to be $z$ and leaves the remaining $n-(t-1)$ sections unfixed.
As $z$ is good, we have that $F^\ast\uhr \rho_{z}$ is nontrivial (i.e., $F^\ast\uhr \rho_{z}\not\equiv 0$ or $1$).
We write $H_1,\ldots,H_m$ (with multiplicities) to denote the set of all depth-$(\ell-1)$ 
  sub-circuits rooted at children of the output gate of $F^\ast$ such that
  $H_i\uhr \rho_z$ is nontrivial. {In other words, we assume that the same sub-circuit may appear multiple times in this list if the output majority gate in $F^\ast$ contains multiple wires to it.}  Since the size of $(F^\ast)$ is at most $M$, the fan-in of the output majority gate of $F^\ast$ is 
  at most $M$, and consequently $m \leq M.$
Since $F^\ast\uhr \rho$ is nontrivial there is a positive integer $h\in [M]$ such that
  $F^\ast\uhr\rho_z$ outputs $1$ if and only if at least $h$ many of 
  $H_1 \uhr\rho_z, \dots, H_m \uhr \rho_z$ output $1$.
The following claim shows that with non-negligible probability, a random $\xx \sim \calD$ is such that ``many'' $H_i$'s become trivial {(i.e., compute a constant function)} after a restriction by $\rho_{z \circ \xx}$:

\begin{claim}\label{claimclaim}
Suppose that $z$ is a good string in the support of $\calD_t$.
Then we have
\[
\Pr_{\xx \sim \calD}\Big[ \big|\{i \in [m]: H_i \uhr \rho_{z \circ \xx} \text{~is trivial~}\}\big| \geq \delta^2 m / 2\Big] \geq
\delta/ {4}.
\]
\end{claim}
\begin{proof}
We consider two cases: $h\ge m/2$ or $h<m/2$.
We focus on the latter below and the former case is symmetric.
Assume that $h<m/2$.
Since $z$ is good, we have
$$
\Pr_{\by\sim \nod_{\ell-1}'}\big[F^\ast(z\circ \by\circ \11)=0\big]\ge 1+6\delta-1=6\delta.
$$
However, if $y\in \supp(\nod_{\ell-1}')$ satisfies
  $F^\ast(z\circ y\circ \11)=0$, then by $h < m/2$
it must be the case that at least $m/2$ of the $H_i$'s have $H_i(z\circ y\circ \11)=0$,
 and hence 
\begin{equation}\label{tutu}
\E_{\by \sim \nod_{\ell-1}'}\big[\hspace{0.03cm}\text{number of $H_i$'s with $H_i(z\circ \by\circ \11)=0$}
\hspace{0.03cm}\big] \geq 3 \delta m.
\end{equation}
Let $I$ denote the set of $i\in [m]$ such that
\begin{equation}\label{dono}
\Pr_{\by\sim \nod_{\ell-1}'}\big[H_i(z\circ \by\circ \11)=0\big]\ge 2\delta.
\end{equation}
Then we have from (\ref{tutu}) that 
$$
|I|\cdot 1+(m-|I|)\cdot 2\delta\ge 3\delta m,
$$
which implies that $|I|\ge \delta m$.  

We write $\rho$ to denote the restriction over $\{0,1\}^{{\ell} \times N_\ell^\ast}$ that fixes the first $ {t-1}$ sections of input variables to be $z$
  and the last $ {(n-t)}$ sections of input variables to be all $1$, and leaves only the variables in the
  $t$-th section unfixed.
So each $H_i\uhr \rho$ is a depth-$(\ell-1)$ monotone majority circuit
  over $\{0,1\}^{ {\ell} \times N_{\ell-1}}$ of size at most $M$.
Then combining (\ref{dono}) and the assumption of the lemma, i.e., (\ref{eq:condcond}), applied to $H_i\uhr \rho$,
  we have that each $i\in I$ satisfies 
$$
\Pr_{\bx\sim \yesd_{\ell-1}'}\big[H_i(z\circ \bx\circ \11)=1\big]\le 1+\delta-2\delta=1-\delta,
$$
and thus,
\begin{equation} \label{eq:explowerbound}
\Pr_{\bx\sim \yesd_{\ell-1}'}\big[H_i(z\circ \bx\circ \11)=0\big]\ge \delta.
\end{equation}

Note that if an $x\in \supp(\yesd_{\ell-1}')$ satisfies
  $H_i(z\circ x\circ \11)=0$,
then we have $H_i\uhr \rho_{z\circ x}\equiv 0$ by the monotonicity of 
  $H_i$. 
Let $\bX$ be a random variable that denotes the number of $H_i$'s that become trivial after 
  $\rho_{z\circ\xx}$, where $\xx\sim \yesd_{\ell-1}'$. So by (\ref{eq:explowerbound})
  the expectation of $\bX$ is at least $\delta |I|$. 
Let $q$ denote the probability that $\bX\ge \delta |I|/2$. The lower bound $\E[\bX] \geq \delta |I|$ implies that 
$$
q\cdot |I|+(1-q)\cdot \delta |I|/2\ge \delta |I|, 
$$
and thus $q\ge \delta/2$.
Plugging in $|I|\ge \delta m$, we have that $\bX\ge \delta^2 m/2$ with
  probability at least $\delta/2$.
  
Finally, taking into account that a draw of $\xx \sim \calD$ is drawn from $ {\yesd_{\ell-1}'}$ with probability $1/2$, we see that with probability at least $\delta/4$ over a draw of $\xx \sim \calD$, we have that at least $\delta^2 m/2$ many $H_i$'s become trivial after $\rho_{z\circ\xx}$.
This finishes the proof of the claim. 
\end{proof}

Claim \ref{claimclaim} implies that if $\zz\sim \calD_t$ is good (with respect to $F^\ast$), then with probability
at least $\delta/4$ over a random draw of $\xx \sim \calD$, the restriction
$\rho_{\zz \circ \xx}$ trivializes at least  $(\delta^2/2)$-fraction of the depth-$(\ell - 1)$ sub-circuits of $F$ that are
\emph{not} trivialized by $\rho_{\zz}$. {Intuitively, this is useful because it means that we have a good chance of getting a significant simplification of $F$ (shrinking the fan-in of the top gate by a lot), and since $F$ is of size at most $M$ this cannot happen too many times.}
On the other hand, note that if $\zz$ is not good, then by definition we have
$$
\Pr_{\bx\sim \yesd_{\ell-1}}\big[F^\ast(\zz\circ \bx\circ \00)=1\big] +\Pr_{\by\sim \nod_{\ell-1}'}
  \big[F^\ast(\zz\circ \by\circ \11)=0\big]< 1+6\hspace{0.02cm}\delta,
$$
which intuitively is also useful for our purpose of bounding 
\begin{equation}\label{sum}
\Pr_{\bx\sim \yesd_{\ell}^*} \big[ F^\ast(\bx)=1\big]+\Pr_{\by\sim \nod_{\ell}^*}\big[F^\ast(\by)=0\big] 
\end{equation}
from above by $1+7\delta$.

To finish the proof of the lemma, we take the following alternative but equivalent view of
  (\ref{sum}).
Let $\bz_1,\ldots,\bz_n$ be a sequence of random section strings, each drawn from $\calD$ independently.  
By the definition of $\yesd_\ell^*$ and $\nod_\ell^*$ (recall Figure~\ref{illustration0}), we have that
\begin{align*}
\text{(\ref{sum})}\times n\hspace{0.06cm}=\hspace{0.06cm}
 \E_{\hspace{0.03cm}\bz_1,\ldots,\bz_n} & \left[\hspace{0.1cm}\sum_{t=1}^{n}\hspace{0.06cm} \Pr_{ \bx\sim \yesd_{\ell-1}} \big[F^\ast(\bz_1\circ \cdots \circ \bz_{t-1}\circ
  \bx\circ \00)=1\big]\right. \\ 
& \ \ \ \ \ \ \left.+\sum_{t=1}^{n}\Pr_{\by\sim \nod_{\ell-1}'}\big[F^\ast(\bz_1\circ\cdots \circ \bz_{t-1}\circ \by\circ \11)=0\big] \hspace{0.05cm}\right].
\end{align*}
This can be viewed as the expectation of a random variable $\bGamma$ generated as follows.
\begin{enumerate}
\item Start with $\bGamma=0$.\vspace{-0.18cm}
\item For each ``round'' $t=1,\dots,n$, independently draw $\zz_t$ from $\calD$ and add the following to $\bGamma$:
$$
\Pr_{\bx\sim \yesd_{\ell-1}}\big[F^\ast(\bz_1\circ \cdots \circ \bz_{t-1}\circ
  \bx\circ \00)=1\big]+\Pr_{\by\sim \nod_{\ell-1}'}\big[F^\ast(\bz_1\circ\cdots \circ \bz_{t-1}\circ \by\circ \11)=0\big].
$$
\end{enumerate}
So it suffices to show that $\E[\bGamma]\le (1+7\delta)n$.

For each of the $n$ rounds $t=1,\dots,n$, exactly one of the following
two possibilities must hold:
\begin{flushleft}
 \begin{enumerate}  
 \item The current string 
  $\bz_1\circ\cdots \circ \bz_{t-1}$ is not good.
In this case $\bGamma$ goes up by at most $1+6\delta$ in the $t$-th round.  Otherwise,

\item  The current string $\bz_1\circ\cdots \circ \bz_{t-1}$ is good.
In this case $\bGamma$ can go up by at most $2$ in the $t$-th round, but by our previous analysis (specifically, Claim \ref{claimclaim}),
  the number of nontrivial depth-$(\ell-1)$ subcircuits of $F^\ast$ (with multiplicities) rooted at children of the output gate of $F^\ast$ drops by a factor of $(1-\delta^2/2)$ with probability at least $\delta/4$
  when the draw of $\zz_t$ in the $t$-th round extends the restriction  to $\rho_{\bz_1 \circ \cdots \circ \bz_{t}}$.
Note that $F^\ast$ has size $M\le 2^{\varepsilon^5n}$ so it can survive at most $ {2}\delta^3 n$ many such drops  
  before $F^\ast$ becomes trivial; to see this, observe that
\begin{equation}\label{eq:int}
(1-\delta^2/2)^{ {2}\delta^3 n}\le \exp\left(-(\delta^2/ {2})\cdot ( {2}\delta^3 n)\right)
=\exp\left(-\delta^5 n\right) < 2^{-\varepsilon^5 n}.
\end{equation}  
Note further that once $F^\ast$ becomes trivial, $\bGamma$ goes up by $1$ in every subsequent round.
\end{enumerate}
\end{flushleft}
  
We use $\bS$, a random variable, to denote the total number of rounds $t \in [n]$ such that the current string
  $\bz_1\circ\cdots\circ\bz_{t-1}$ is good
  (note that once $F^\ast$ becomes trivial the current string cannot be good).
We claim that $\bS\le 32\delta^2 n$ with high probability.

\begin{claim}\label{martingale}
We have $\bS\le 32\delta^2 n$ with probability at least 
  $1-\exp(-n\delta^4 /2 )$.
\end{claim}
\begin{proof}
We say that round $t$ is good if the current string $\bz_1\circ\cdots \circ\bz_{t-1}$ is good.
We say that $F^\ast$ is \emph{hit} in the $t$-th round, if $\bz_1\circ\cdots\circ\bz_{t-1}$
  is good and the number of depth-$(\ell-1)$ subcircuitscuits of $F^\ast$ (with multiplicities) that are trivial under the restriction
  $\smash{\rho_{\bz_1\circ\cdots\circ\bz_{t-1}}}$
  drops by a factor of at least  $(1-\delta^2/2)$ under the restriction $\rho_{\bz_1\circ\cdots\circ \bz_{t-1}\circ\bz_{t}}$.
Then we can write $\Pr\hspace{0.03cm}[\bS \geq 32 {\delta}^2 n ]$ as 
\begin{align*}
&\Pr\big[\hspace{0.05cm}\bS \geq 32 {\delta}^2 n\ \&\ \text{$F^\ast$ is 
  hit $>  {2 \delta}^3n$ many times during the first $32 {\delta}^2 n $ of the good rounds}\hspace{0.05cm}\big]\\[0.3ex]&+
  \Pr\big[\hspace{0.05cm}\bS \geq 32 {\delta}^2 n\ \&\ \text{$F^\ast$ is hit $\le  {2\delta}^3 n $ many times during the first $32 {\delta}^2 n $ of the good rounds}\big].
\end{align*} 
The first of these probabilities is zero because of (\ref{eq:int}), i.e.  
  if $F^\ast$ is hit $ {2}\delta^3 n$ times then it is trivialized so no subsequent rounds can be good
  and thus $F^\ast$ cannot be hit again.

We focus on upper bounding the second probability. For each $i$ from $1$ to $32\delta^2 n$ we define the
  following random variable $\bY_i$ where
\begin{equation*}
\bY_i=\begin{cases}1 & \text{if $F^\ast$ is hit in the $i$-th good round or there 
  are fewer than $i$ good rounds}\\
0 & \text{otherwise (there are at least $i$ good rounds and $F^\ast$ is not hit in
  the $i$th good round).}
\end{cases}
\end{equation*}
The second probability we are interested in is at {most}
 $\Pr\hspace{0.03cm}[\sum_i \bY_i\le {2}\delta^3 n]$.
By Claim \ref{claimclaim}, we have
\begin{equation} \label{eq:yineq}
\Ex \big[\hspace{0.05cm}\bY_i\hspace{0.08cm}|\hspace{0.08cm}\bY_1=b_1,\ldots, \bY_{i-1}=b_{i-1}
\hspace{0.05cm}\big]\ge \delta/4
\end{equation}
for all $i$ and all $b_1,\ldots,b_{i-1}\in \{0,1\}$.
Let $\bX_0\equiv 0$ and $$\bX_i=\bX_{i-1}+\bY_i-\Ex \big[\hspace{0.05cm}\bY_i
\hspace{0.08cm} |\hspace{0.08cm} \bY_1 ,\cdots, \bY_{i-1} \hspace{0.05cm}\big].$$
Then $\bX_0,\bX_1,\ldots$ is a martingale that satisfies 
$|\bX_{i}-\bX_{i-1}|\le 1$
with probability $1$, {and we have that
\[
\bX_{32\delta^2 n} = \sum_{i=1}^{32 \delta^2 n} \left( \bY_i -\Ex \big[\hspace{0.05cm}\bY_i
\hspace{0.08cm} |\hspace{0.08cm} \bY_1 ,\cdots, \bY_{i-1} \hspace{0.05cm}\big]\right) 
\leq
 \sum_{i=1}^{32 \delta^2 n} \bY_i  - 8\delta^3 n,
\]
using (\ref{eq:yineq}) for the inequality.}
Applying the Azuma-Hoeffding inequality (see, e.g., Theorem 5.1 of \cite{DP09}) to
the martingale $\bX_0,\bX_1,\ldots$, we get that
$$
\Pr\left[ \sum_i \bY_i\le {2}\delta^3 n\right]\le \Pr\left[
\hspace{0.05cm}\bX_{ {32\delta^2 n}}\le {2}\delta^3n-8\delta^3 n\hspace{0.05cm}
\right]\le \exp\left(-\frac{ (6\delta^3 n)^2}{2\cdot 32\delta^2 n}\right)  {<}
\exp(- {n} \delta^4 {/2}).
$$
This finishes the proof of the claim.
\end{proof}

We are almost done with the proof of Lemma~\ref{l:inductive_step}.  Recalling that $\delta=7^{\ell- {2}}\epsilon$, we have that
\begin{equation}\label{haha}
\exp\left(- {n}\delta^4/2\right)\le \delta/4\quad\text{and}\quad \delta\le 2^{-8} 
\end{equation}
since $d\ge {2}$, $n\ge 2^{60d}$, $\eps = 2^{-12d}$ and $\ell \in  {\{2, \ldots, d\}}.$
It follows from Claim \ref{martingale} that 
\begin{align*}
\E\hspace{-0.05cm}\big[\bGamma\big]&\le \exp(- {n}\delta^4 /2)\cdot 2n+(1-\exp(- {n}\delta^4  {/2}))\cdot \left( 2\cdot 32\delta^2 n +
  (1+6\delta)\cdot (n-32\delta^2 n )\right)\\[0.3ex]
&<\delta n/2 + 64\delta^2 n+(1+6\delta)n\le (1+7\delta)n,
\end{align*}
where we also used the two inequalities in (\ref{haha}).
This finishes the proof of the lemma.
\end{proof}

\subsection{Proof of Theorem \ref{thm:lower}.}\label{sec:proof}

Finally we combine all the ingredients to prove Theorem \ref{thm:lower}.

Recall that $d\ge 2$, $n$ and $N$ are positive integers that satisfy $n\ge 2^{60d}$ and $N\ge {(2^{13}}n)^d\ge N_d$.
We also have $\epsilon=2^{-12d}$ and $M=2^{\epsilon^5 n}$.
We first prove by induction on $\ell$ that, for $\ell= {1},\ldots,d$, any
  monotone majority circuit $F$ over $\{0,1\}^{ {(\ell+1)} \times N_\ell}$ 
  of depth $\ell$ and size at most $M$ satisfies
\begin{align}
&\Pr_{\bx\sim \yesd_{\ell }} \big[ F(\bx)=1\big]+\Pr_{\by\sim \nod_{\ell }}\big[F(\by)=0\big]\le 1+ 7^{ {\ell-1}}
\epsilon,\quad\text{and}\nonumber\\[0.5ex] &
\Pr_{\bx\sim \yesd_{\ell }'} \big[ F(\bx)=1\big]+ \Pr_{\by\sim \nod_{\ell }'}\big[F(\by)=0\big]\le 1+
7^{ {\ell-1}}\epsilon.
\label{eq:condcond2}
\end{align}

The $\ell= {1}$ base case follows from Lemma \ref{lem:yes0no0}.
Now assume that (\ref{eq:condcond2}) holds for $\ell-1$.
By Lemma \ref{l:inductive_step}, any monotone majority circuit $F^\ast$ over $\{0,1\}^{ {\ell} \times N_\ell^*}$
  of depth $\ell$ and size at most $M$ satisfies
\begin{equation} \label{eq:star}
\Pr_{\bx\sim \yesd_{\ell}^*} \big[ F^\ast(\bx)=1\big]+\Pr_{\by\sim \nod_{\ell}^*}\big[F^\ast(\by)=0\big]\le 1+7^{ {\ell} - 1} \epsilon.
\end{equation}
It follows from Lemmas \ref{lem:trivial1} and \ref{lem:trivial2} that 
  every monotone majority circuit $F$ over $\{0,1\}^{ {(\ell + 1)} \times N_\ell}$ of 
  depth $\ell$ and size at most $M$ satisfies (\ref{eq:condcond2}).
This finishes the induction.

We finish the proof using $(\yesd_d^{\ast},\nod_d^{\ast})$ over $\{0,1\}^{d\times N_{d}^*}$, where $N_d^*=N_d-1< N$.
Given (\ref{eq:star}) on $(\yesd_d^{\ast},\nod_d^{\ast})$ and the fact that $7^{d-1}\epsilon < 1$, no depth-$d$ monotone 
  majority circuit on $\{0,1\}^{ {d} \times N_d^{\ast}}$ of size at most $M$
  can compute $U_{ {d}, N^\ast_d}$ correctly on all inputs, because every string $\xx\sim \yesd_d^{\ast}$ has $\smash{\SUM(\xx)=2^{N_d^{\ast}}}$ and hence $\smash{U_{ {d}, N_d^\ast}(\xx)=1}$, while 
  every string $\yy\sim\nod_d^{\ast}$ has $\smash{\SUM(\yy)=2^{N_d^\ast}-{d}}$ and hence $U_{ {d}, N_d^{\ast}}(\yy)=0.$
Since $N>N_d^*$, this establishes Theorem \ref{thm:lower}.

% %%%%%%%%%%%%%%%%%%%%%%%%%%%%%%%%
% upper_bounds
% %%%%%%%%%%%%%%%%%%%%%%%%%%%%%%%%

\section{The Upper Bound: Proof of Theorem \ref{thm:upper}.}\label{s:upper_bounds}

We prove Theorem \ref{thm:upper} in this section. 
We focus on the case when $N^{1/d}>1$ is a positive integer, and 
  give a depth-$d$ monotone majority circuit that computes $U_{k,N}$ and has size at most
\begin{equation}\label{lala}
  {2^{3(N^{1/d}\hspace{0.02cm}\cdot\hspace{0.02cm} \log k+\log N)}}.
\end{equation}
For the general case, we let $n=\lceil N^{1/d}\rceil>1$, and 
  let $s$ denote the smallest integer such that $n^s\ge N$ (so $s\le d$).
Then we first construct a depth-$s$ monotone majority circuit
  that computes $U_{k,n^s}$, and then hard-wire the variables in the last $n^s-N$
  columns to be $0$ to get a circuit for $U_{k,N}$.
The size bound given in the statement of Theorem \ref{thm:upper}
  follows from (\ref{lala}) and the simple facts that $n\le 2N^{1/d}$ 
  and $n^s\le nN\le N^2$.
For the rest of the section we assume that $n=N^{1/d}>1$ is an integer.

First we note that the theorem (with the size bound as given in (\ref{lala}); the same below)
  is trivial if $N<\log k$ since
  implementing $U_{k,N}$ directly using a single $\THR$ gate only takes
  a total weight of $k\cdot 2^N< 2^{3\log k}$.
Assuming that $N\ge \log k$ below, we let $t\in \{1,\ldots,d\}$ denote the 
  smallest integer such that
  $n^t=N^{t/d} \ge \log k.$
We also write $M=n^t$.
It is clear by the choice of $t$ that we have
\begin{equation}\label{hehe}
M\le n \log k.
\end{equation}
With the same reasoning the theorem is trivial if $M=N$.
Below we assume that $t\le d-1$.

We need some notation for our construction.
We say $\mathcal{S} = (S_1, \ldots, S_\ell)$ is an $\ell$-\emph{decomposition} of $[N]$ if there exist
  indices $1 = a_1^- \leq a_1^+ < a_2^- \leq a_2^+ < \ldots < a_\ell^- \leq a_\ell^+ = N$ such that
\begin{itemize}
\item For each $\gamma \in [\ell]$, $S_\gamma = \{a_\gamma^-, a_\gamma^- + 1, \ldots,  a_\gamma^+\}$; and
\vspace{-0.1cm}
\item $\bigcup_{\gamma \in [\ell]} S_\gamma = [N]$.
\end{itemize}
In other words, $\mathcal{S}$ partitions $[N]$ into $\ell$ sequential intervals.

Let $\smash{(x_{i,j})_{i \in [k], j \in [N]}}$ be the set of input variables of $U_{k,N}$. 
Given an $\ell$-decomposition $\calS$,
  we define a sequence of ``conditional'' \emph{carry-bit functions} 
  $\smash{c^{(\gamma)}_{\alpha,\beta}}(x)$, where $0 \leq \alpha,\beta \leq k-1$ and $\gamma \in [\ell]$. Each function $\smash{c^{(\gamma)}_{\alpha,\beta}}$ depends only on the variables $x_{i,j}$ with $j\in S_\gamma$. For convenience, let $B_\gamma = [k] \times S_\gamma$ be the set containing the indices of these variables. Intuitively, for an assignment $x \in \{0,1\}^{k\times N}$, we have 
  $\smash{c^{(\gamma)}_{\alpha,\beta}(x) = 1}$ if and only if {a carry of value at least $\alpha$ is} generated/propagated by the input bits corresponding to $B_\gamma$, assuming this block of variables receives {a carry of value $\beta$}  from the block of variables to the right. Formally,
\begin{equation}\label{eq:base_case}
c^{(\gamma)}_{\alpha,\beta}(x) \eqdef 1 \quad \Longleftrightarrow \quad \sum_{(i,j) \in B_\gamma} 2^{|S_{\gamma}| - (j +1 - a_\gamma^-)} \cdot x_{i,j} + \beta \geq \alpha \cdot 2^{|S_\gamma|}.
\end{equation}
For each $i\in \{0,\ldots,d-t\}$, we write $\calS^{(i)}$ to denote 
  the $n^i$-decomposition in which each set has size $n^{d-i}  $. 
Observe that, for the $1$-decomposition $\smash{\mathcal{S}^{(0)} = (S^{(0)}_1)}$, 
  where $\smash{S^{(0)}_1 = \{1, \ldots, N\}}$, we have  
\begin{equation}\label{eq:correctness}
U_{k,N}(x) = 1 \quad \Longleftrightarrow \quad c^{(1)}_{1,0}(x) = 1,
\end{equation}
for the function $c_{1,0}^{(1)}$ of $\calS^{(0)}$.

Our construction is based on a recursive computation of functions 
  $\smash{c^{(\gamma)}_{\alpha,\beta}(\cdot)}$ associated to different decompositions
  $\calS^{(r)}$, for $r$ from $d-t$ back to $0$, where each decomposition $\calS^{(r+1)}$ is obtained via~a refinement 
  of the previous decomposition $\calS^{(r)}$. 
  More precisely we construct our monotone~majority circuit for $U_{k,N}$ with the following intended behavior. The top gate of the circuit computes the bit $\smash{c^{(1)}_{1,0}(x)}$ associated to the decomposition $\mathcal{S}^{(0)}$. However, this gate does not have access to $x$: it receives as input the output of carry-bit functions $c^{(\gamma)}_{\alpha, \beta}(x)$ corresponding to the finer $n$-decomposition $\smash{\calS^{(1)}}$ in which 
   each block has $\smash{n^{d-1}}$ columns. This then leads to a recursive procedure, which unfolds as a 
   depth-$(d-t+1)$ circuit described in more detail below (recall that $t\ge 1$).

In general  our circuit has $d-t+1$ layers of majority gates, where gates at the $i$th layer compute carry-bit functions  $\smash{c_{\alpha,\beta}^{(\ell)}}$ corresponding to the $n^{ d-t-i+1 }$-decomposition $\calS^{(d-t-i+1)}$. The base~case, i.e. the first layer of majority gates 
  that are supposed to compute $\smash{c_{\alpha,\beta}^{(\ell)}}$ of $\calS^{d-t}$,
  is done by a majority gate that follows directly the definition given in 
  (\ref{eq:base_case}).
It is clear that the size of each gate in the first layer is bounded from above by $k 2^M$.

Due to the recursive nature of our construction, it is sufficient to describe how to compute the carry-bit functions corresponding to a decomposition $\mathcal{S}^{(r)}$ from the carry-bit functions corresponding to $\mathcal{S}^{(r+1)}$ for each $r\in \{0,1,\ldots,d-t-1\}$. For convenience we fix an $r$ below and write $\mathcal{S}'$
  for $\mathcal{S}^{(r)}$ and $\mathcal{S}$ for $\mathcal{S}^{(r+1)}$. We also fix a set $S' \in \mathcal{S}'$ with $S' = S_1 \cup \ldots \cup S_n$,
  where $S_1,\ldots,S_n$ are sets in the ordered tuple $\calS$ listed from left to right. 
We write $c_{u,v}$ to denote a carry-bit function of the block $S$ that 
  we need to compute, for some $u,v\in \{0,\ldots,k-1\}$, and 
  assume that we have already computed $\smash{c^{(\gamma)}_{\alpha,\beta}}$ for each block $S_\gamma$, $\gamma \in [n]$, and for all $\alpha, \beta \in \{0, \ldots,k-1\}$. 
The goal is to compute $c_{u,v}(x)$ given the bits $\smash{c^{(\gamma)}_{\alpha, \beta}(x)}$. 
  
We start with a general observation about carry-bit functions of a block. We say that $(\alpha,\beta)\prec (\alpha',\beta')$ if either $\alpha<\alpha'$, or $\alpha=\alpha'$ and
  $\beta\ge \beta'$. Given a block $S_\gamma$, note that $\smash{c_{\alpha,\beta}^{(\gamma)}}$
  has the following monotonicity property. 
(Note that the assumption of $|S_\gamma|\ge \log k$ always holds given our choice of $t$ and 
  trivial cases ruled out at the beginning of the section.)

\begin{claim}\label{c:carry_monotonicity} Assume that $|S_\gamma| \geq \log k$. If $(\alpha,\beta)\prec (\alpha',\beta')$, then
  $c_{\alpha,\beta}^{(\gamma)}(x) \ge c_{\alpha',\beta'}^{(\gamma)}(x)$ on every input string $x$ for $U_{k,N}$.
\end{claim}

\begin{proof}
We consider the two cases corresponding to the assumption that $(\alpha,\beta)\prec (\alpha',\beta')$. If $\alpha = \alpha'$ and $\beta \geq \beta'$, the claim follows immediately from (\ref{eq:base_case}). 

Assume now that $\alpha < \alpha'$, where $\beta, \beta' \in \{0, \ldots, k-1\}$ are arbitrary. Clearly it suffices to argue that $\smash{c^{(\gamma)}_{\alpha',k-1}(x) = 1}$ implies that $\smash{c^{(\gamma)}_{\alpha'-1,0}(x) = 1}$. Using (\ref{eq:base_case}), this assumption is equivalent to
\begin{equation}\label{eq:proof_claim}
\sum_{(i,j) \in B_\gamma} 2^{|S_{\gamma}| - (j +1 - a_\gamma^-)} \cdot x_{i,j} + (k-1) \geq \alpha' \cdot 2^{|S_\gamma|}.
\end{equation} 
In order to show $c^{(\gamma)}_{\alpha'-1,0}(x) = 1$, we need to verify that 
$$
\sum_{(i,j) \in B_\gamma} 2^{|S_{\gamma}| - (j +1 - a_\gamma^-)} \cdot x_{i,j} \geq (\alpha'-1) \cdot 2^{|S_\gamma|}.
$$
Using (\ref{eq:proof_claim}) it is sufficient to have $k-1 \leq 2^{|S_\gamma|}$. This follows from the assumption in the statement of the claim, which completes the proof.
\end{proof}

The description of the majority gate that computes $c_{u,v}(x)$ 
  for the block $S'$ in $\calS'$ using $c_{\alpha,\beta}^{(\gamma)}(x)$ for
  blocks $S_1,\ldots,S_n$ in $\calS$ is based on the following lemma. 

\begin{lemma} \label{l:inductive_construction} Assume that $|S_\gamma| \geq \log k$ for every $\gamma \in [n]$. Then 
  $c_{u,v}(x)=1$ if and only if
\begin{equation}\label{eq:sum1}
v+\sum_{\gamma=1}^n\left(\left( \sum_{ \alpha = 1,\beta=0}^{k-1} c^{(\gamma)}_{\alpha,\beta}(x)\right) \cdot 
  k^{n-\gamma}\right)\ge u \cdot k^n.
\end{equation}
\end{lemma}

\begin{proof} 
We consider (\ref{eq:sum1}) as a sum in base $k$ over $k(k-1)$ rows and $n$ columns of 
  variables, with $v$ extra $1$'s on column $n$ (which corresponds to the least significant position). 
Let $p_\gamma$ denote the (base $k$) carry from column $\gamma$ to column $\gamma-1$ in (\ref{eq:sum1}),
  and let $q_\gamma$ denote the (base $2$) carry from block $\gamma$ to block $\gamma-1$ in our decomposition of $x$ after adding $v$ to
  block $n$ (without taking into account the remaining columns of $x$ not covered by $S_1 \cup \ldots \cup S_n$). 
  
  We prove by induction that $p_\gamma=q_\gamma$, for all $\gamma$ from $n$ to $1$. Notice that this establishes the lemma. For the basis when $\gamma=n$, we consider the following two cases: 
\begin{flushleft}\begin{enumerate}
\item If $c_{\alpha,\beta}^{(n)} = 0$ for all $\alpha \geq 1$ and $\beta \geq 0$, then $q_{n}=0$
  (as we have $c_{1,k-1}^{(n)}=0$ and $v\le k-1$).
This implies that $p_n=q_n=0$.
  
\item Otherwise, let $(\alpha_n,\beta_n)$ denote the largest pair (under $\prec$ defined earlier)
  with $c_{\alpha_n,\beta_n}^{(n)}=1$. It follows from Claim \ref{c:carry_monotonicity} that
  $q_n = \alpha_n$ if $\beta_n\le v$, and $q_n = \alpha_n-1$ if $\beta_n>v$.
We also have
$$
v+\sum_{\alpha = 1,\beta= 0}^{k-1} c^{(n)}_{\alpha,\beta}= (\alpha_n-1) \cdot k+(k-\beta_n+v).
$$
It follows from this equation and the characterization of $q_n$ that the (base $k$) carry $p_n=q_n$.
\end{enumerate}\end{flushleft}

The induction step is similar. We assume that $p_{\gamma+1}=q_{\gamma+1}$, and prove that $p_{\gamma}=q_{\gamma}$.
We focus on the $\gamma$-th column from (\ref{eq:sum1}) and block $\gamma$, and consider the following two cases:
\begin{flushleft}\begin{enumerate}
\item If $c_{\alpha,\beta}^{(n)} = 0$ for all $\alpha \geq 1$ and $\beta \geq 0$, then $q_\gamma=0$
  (as we have $c_{1,k-1}^{(\gamma)}=0$ and $q_{\gamma+1}\le k-1$).
This implies that $p_\gamma=q_\gamma=0$.
\item Otherwise, let $(\alpha_\gamma,\beta_\gamma)$ denote the largest pair 
  with $c_{\alpha_\gamma,\beta_\gamma}^{(\gamma)}=1$. Using Claim \ref{c:carry_monotonicity}, $q_\gamma$ is $\alpha_\gamma$ if $\beta_\gamma\le q_{\gamma+1}$, and $q_\gamma$ 
  is $\alpha_\gamma-1$ if $\beta_\gamma>q_{\gamma+1}$.
Using the inductive hypothesis, we have
$$
p_{\gamma+1}+\sum_{\alpha = 1,\beta= 0}^{k-1} c^{(n)}_{\alpha,\beta}= (\alpha_\gamma-1) \cdot k+(k-\beta_\gamma+q_{\gamma+1}).
$$
It follows from this equation and the characterization of $q_\gamma$ that $p_\gamma=q_\gamma$.
\end{enumerate}\end{flushleft}
This finishes the induction, and the proof of the lemma.
\end{proof}

Lemma \ref{l:inductive_construction}, (\ref{eq:base_case}), and our previous discussions complete the description of the circuit for $U_{k,N}$. Moreover, its correctness follows easily from (\ref{eq:correctness}) and Lemma \ref{l:inductive_construction}. It remains to analyze the size of the resulting depth-$(d-t+1)$ majority circuit. 

We upper bound its size layer by layer as follows. 
As discussed earlier, the size of each majority gate in the first layer 
  is at most $k2^M$, and there are $n^{d-t}$ many of them.
Furthermore, for the $i$-th layer of the circuit, where $i>1$, there are $n^{(d-t-i+1)}$ gates
  each of which has size at most $$k(k-1)\cdot \frac{k^n-1}{k-1}<k^{n+1},$$  as given in Lemma \ref{l:inductive_construction}. Using (\ref{hehe}), the majority circuit for $U_{k,N}$ has overall size at most 
$$
n^{ d-t } \cdot k2^M+\sum_{i=2}^{d-t+1} n^{d-t+1-i } \cdot k^{n+1}\le
Nk2^M+2Nk^{n+1}\le 
  2^{3(N^{1/d} \log k + \log N)}.
$$

The construction presented here uses $\THR$ gates and majority circuits.
We sketch in Appendix \ref{s:connectivity} an alternative construction 
  with respect to semi-unbounded fan-in $\mathsf{AND}$/$\mathsf{OR}$ circuits.

% %%%%%%%%%%%%%%%%%%%%%%%%%%%%%%%%
% exponential_gap
% %%%%%%%%%%%%%%%%%%%%%%%%%%%%%%%%

\section{Strengthening the Ajtai-Gurevich Result: Proof of Theorem \ref{t:strong_AG}.}\label{s:strong_AG}

We require the following  lemma:

\begin{lemma} \label{lem:add}
For a suitable absolute constant $0 < c < 1$, letting $k= (\log N)^c$, the function $U_{k,N}$ is computed by a poly$(N)$-size  $\AND/\OR/\NOT$
circuit of depth \emph{3}.
\end{lemma}
\begin{proof}  Recall the well-known technique of carry-save addition, also known as the ``3-to-2 trick,'' for addition of binary numbers (see e.g., Section 1.2.3 of \cite{Leighton:92}).  This ``trick'' states that there is a (multi-output) circuit that takes as input three $n$-bit binary numbers $X,Y,Z$ and outputs two $(n+1)$-bit binary numbers $A,B$ such that (\emph{i}) $A+B=X+Y+Z$, and (\emph{ii}) each output bit $A_i$ or $B_i$ depends on at most 3 of the input bits.  By applying this trick in parallel to the $N$-bit integers $x^{(1)},\dots,x^{(k)}$ that are the rows of the input to $U_{k,N}$, we obtain $\lceil 2k/3 \rceil$ many $(N+1)$-bit integers whose sum equals $x^{(1)} + \cdots + x^{(k)}$.  Recursing $O(\log k)$ times, we see that there are two $(N+O(\log  k))$-bit integers (call them $y$ and $z$) such that $y+z= x^{(1)} + \cdots + x^{(k)}$.  A naive composition of these ``3-to-2 trick'' circuits in a tree of depth $O(\log  k)$ to compute $y,z$ would yield a circuit of depth $\Theta(\log \log N)$.  To avoid this blowup in circuit depth, we proceed differently, by observing that each each bit $y_i,z_i$ depends on at most $3^{O(\log  k)} \leq  \log N$ of the original input bits of the $x^{(i)}$'s, and exploiting this locality to get a depth-3 circuit overall.

In more detail, let $y_i$ denote the bit in the ``$2^i$-position'' of the binary representation of $y$, so
\[
y = \sum_{i=0}^{N + O(\log k)} y_i \cdot 2^i \quad\ \  \text{and similarly} \ \ \quad
z = \sum_{i=0}^{N + O(\log k)} z_i \cdot 2^i.
\]
We define ``generate'' and ``propagate'' bits for each bit position of $y+z$ in the standard way,
\[
g_i \eqdef y_i \wedge z_i \quad \text{and} \quad p_i \eqdef y_i \vee z_i,
\]
so $g_i=1$ iff the bits in the $2^i$-position generate a carry into the $2^{i+1}$-position, and $p_i=1$ iff the bits in the $2^i$-position propagate an incoming carry into the $2^i$-position onward to the $2^{i+1}$-position.  Observe that each $p_i,g_i$ depends on at most $2\log N$ of the original input bits.

The sum $y+z$ is at least $2^N$ if and only if either of the following events hold:

\begin{itemize}
\item Event $A$:  at least one of the bits $y_N,y_{N+1},\dots,z_N,z_{N+1},\ldots$ is 1.  This can be expressed as
\[
A = \bigvee_{j \geq N} (y_j \vee z_j).
\]
Since $y_j,z_j$ each depend on at most $\log N$ of the original input variables, each of them can be expressed as a $\poly(N)$-size DNF over the original input variables, and thus $A$ can be expressed as a $\poly(N)$-size DNF.
\item Event $B$:  a carry bit is propagated into the $2^N$-position.
Event $B$ can be expressed as
\[
B = \bigvee_{j=1}^{N-1} \left(g_j \wedge \left( \bigwedge_{j < i < N}p_i\right)\right).
\]
As each $p_i$ depends on at most $2\log N$ of the original input variables, it can be expressed~as a $\poly(N)$-size CNF; the same holds for $g_j$, so $$\left(g_j \wedge \left( \bigwedge_{j < i < N}p_i\right)\right)$$ can be expressed as a $\poly(N)$-size CNF, and thus Event $B$ can be expressed as a $\poly(N)$-size depth-3 $\OR$-$\AND$-$\OR$ circuit.

\end{itemize}

As a consequence, $A \vee B$ can be expressed as a $\poly(N)$-size depth-3 $\OR$-$\AND$-$\OR$ circuit over the original input variables, and the lemma is proved.
\end{proof}

\noindent \emph{Proof of Theorem \ref{t:strong_AG}:}
We take $g_{N} \eqdef U_{k,N}$, where $k=(\log N)^c$ as in Lemma \ref{lem:add}. Then Part (\emph{i}) of the theorem follows from Lemma  \ref{lem:add}.  Part (\emph{ii}) follows from our main lower bound, Theorem \ref{thm:lower}, by observing that any circuit for $U_{k,N}$ yields a circuit for $U_{ {d},N}$ (by setting the last $ {k - d}$ rows of the input to $0$).
\qed

\bibliographystyle{alpha}

% %%%%%%%%%%%%%%%%%%%%%%%%%%%
% bibliography
% %%%%%%%%%%%%%%%%%%%%%%%%%%%

\appendix

% %%%%%%%%%%%%%%%%%%%%%%%%%%%
% connectivity
% %%%%%%%%%%%%%%%%%%%%%%%%%%%

\section{Upper Bound for the Universal Monotone Threshold Gate.}\label{s:connectivity}

We sketch in this section a construction of monotone circuits for the universal monotone threshold function that matches the parameters obtained by Beimel and Weinreb \cite{DBLP:conf/coco/BeimelW05}. More precisely, we describe a polynomial size $O(\log N)$-depth $\mathsf{AND}$/$\mathsf{OR}$ circuit for $U_{O(N), O(N\log N)}$, where $\mathsf{OR}$ gates have unbounded fan-in, while $\mathsf{AND}$ gates have fan-in two.

Our construction relies on a more general reduction from $U_{k,N}$ to a certain graph connectivity problem. We start with an $\ell$-decomposition $\mathcal{S}$ of $U_{k,N}$ (see Section \ref{s:upper_bounds} for more details), and assume (for now) that we are given the corresponding (conditional) carry-bit functions $c^{(\gamma)}_{\alpha, \beta}(x)$, where $\alpha$ and $\beta$ are in  $\{0, \ldots, k - 1\}$, and $\gamma \in [\ell]$.

Given these bits, we can view them as a layered directed graph $G_{\mathcal{S},x} = (V,E)$ which depends on $x$ and $\mathcal{S}$ as follows. The vertices of $G$ are partitioned into $\ell + 1$ layers, which we number for convenience from $\ell$ to $0$. The first and last layers are special, and contain a single vertex only. The remaining layers each contain $k$ vertices. The (directed) edges of this graph leave the $\gamma$-th layer and reach the $(\gamma -1)$-th layer. We use the output bit of each function $c^{(\gamma)}_{\alpha,\beta}$ to decide whether an edge is present in this graph. The idea is that there will be a path from the $\ell$-th layer to the $0$-th layer if and only if $U_{k,N}(x) = 1$.   

More precisely, we view $V = L_\ell \cup L_{\ell - 1} \cup \ldots \cup L_0$, where $L_\ell = \{s\}$, $L_0 = \{t\}$, and $L_\gamma = \{v_{\gamma,0}, \ldots, v_{\gamma,k-1}\}$, for $\ell > \gamma > 0$. The edge set $E \subseteq V \times V$ is defined as follows.
\begin{itemize}
\item $(s,v_{\ell - 1, j}) \in E$ if and only if $c^{(\ell)}_{1,j} = 1$, where $j \in \{0, \ldots, k-1\}$;
\item $(v_{1,j}, t) \in E$ if and only if $c^{(1)}_{j,0} = 1$, where $j \in \{0, \ldots, k-1\}$;
\item For $\ell - 1 \geq \gamma \geq 2 $ and $0 \leq \alpha, \beta \leq k-1$, $(v_{\gamma, \alpha}, v_{\gamma - 1, \beta}) \in E$ if and only if $c^{(\gamma)}_{\alpha, \beta} = 1$;
\item There is no other edge in $E$.
\end{itemize}

Given vertices $u,v$ in a graph $G$, we write $u \leadsto v$ if there exists a directed path from $u$ to $v$ in $G$. Our construction is based on the following observation.
 
\begin{lemma}\label{l:path_characterization} Given an $\ell$-decomposition $\mathcal{S}$ for $U_{k,N}$ and an input $x$,
$$
U_{k,N}(x) = 1 \quad \Longleftrightarrow \quad s \leadsto t~\textit{in}~G_{\mathcal{S},x}.
$$
\end{lemma}

\begin{proof} We provide a sketch of the argument. If $U_{k,N}(x)= 1$, consider the sequence of carries generated during the actual computation of $\sum_{i \in [k]} x^{(i)}$ by the standard binary addition algorithm. At least one final carry is generated in this process, since the sum is at least $2^N$. The correct carry values computed during intermediate steps of the addition algorithm correspond to a path from $s$ to $t$ in $G_{\mathcal{S},x}$. On the other hand, if there exists a path from $s$ to $t$ in this graph, then an inductive argument starting from $t$ and proceeding backwards to $s$ shows that, during each step of the addition algorithm, at least some number of carries must be produced when we add the integers $x^{(1)}, \ldots, x^{(k)}$. In particular, there must be at least one final carry bit, which implies that $U_{k,N}(x)= 1$. 
\end{proof}

To sum up, in order to compute $U_{k,N}$ from the carry-bit functions it is enough to solve a directed $s$-$t$-connectivity problem on a graph with $O(N)$ layers, where each layer contains $O(k)$ vertices. 

 The computation of the carry-bit functions can be done efficiently in the case of the universal monotone threshold function if we start with an $\Omega(N \log N)$-decomposition. More precisely, each such function can be written as a monotone majority gate over a polynomial number of input bits, which is known to admit efficient monotone circuits as needed in our construction. 
 
Finally, the upper bound follows from the well-known construction of monotone circuits for $s$-$t$-connectivity on layered graphs via divide-and-conquer.

\end{document}